\documentclass[11pt]{article}
\usepackage[a4paper,margin=25mm]{geometry}
\usepackage[T1]{fontenc}
\usepackage{amsmath,amssymb,amsthm,mathtools}
\usepackage{cite,booktabs}
\usepackage[colorlinks=true,linkcolor=blue,citecolor=blue,urlcolor=blue]{hyperref}
\renewcommand{\thesection}{\Roman{section}}

\newtheorem{theorem}{Theorem}
\newtheorem{lemma}[theorem]{Lemma}
\newtheorem{proposition}[theorem]{Proposition}
\newtheorem{corollary}[theorem]{Corollary}
\theoremstyle{definition}

\theoremstyle{remark}
\newtheorem{remark}[theorem]{Remark}
\newcommand{\F}{\mathbb F}
\newcommand{\Pone}{\mathbb P^1}
\newcommand{\Zpos}{\mathbb Z_{>0}}
\newcommand{\pos}[1]{\left(#1\right)_{+}}
\DeclareMathOperator{\Span}{span}

\DeclareMathOperator{\lcm}{lcm}
\DeclareMathOperator{\PGL}{PGL}
\DeclareMathOperator{\RS}{RS}
\title{Minimum Extension Degrees for Reed--Solomon Codes:\\
Exact Values and a Quadratic Transition}
\author{Li Yuhan\thanks{The Chinese University of Hong Kong, Shenzhen, China. Email: \texttt{yuhanli@link.cuhk.edu.cn}.}}
\date{}

\begin{document}
\maketitle

\begin{abstract}
We study the minimum extension degree of Reed--Solomon codes
admitting exact linear repair of every failed coordinate from
any $d$ surviving coordinates at the cut-set bound.
Write $s=d-k+1$, and let $Q_s(n)$ be the minimum product of
$n$ pairwise coprime integers, each at least $s$.
We locate the transition for the formula $sQ_s(n)$ near
$n=d^2$: for every fixed $0<\varepsilon<1$ and sufficiently large
$d$, the formula is exact up to $(1-\varepsilon)d^2$, whereas
strict improvement is possible at every length at least
$(1+\varepsilon)d^2$, uniformly over the base field and code dimension.
The lower bound follows from an inequality for the cost of residual
prime powers; the upper bound follows from an existing affine-plane
construction. We also show that the full-helper coordinate-count
bound is attained over sufficiently large base fields.
Finally, nonlinear repair at cut-set
equality induces linear repair over the prime field. Thus the
degree bounds over that field remain valid, and the same transition
holds when the repair base field is prime.
\end{abstract}


\section{Introduction}\label{sec:intro}

\subsection{Repair of Reed--Solomon Codes at the Cut-Set Bound}

Maximum distance separable (MDS) codes provide storage redundancy
while allowing the original data to be reconstructed from any $k$
of the $n$ stored symbols. Repairing one missing symbol by first
reconstructing the data, however, generally downloads much more than
the size of that symbol. The regenerating-code framework of
Dimakis et al.~\cite{DGW} quantifies the tradeoff between storage and
repair bandwidth. At the minimum-storage point, a node storing
$\ell$ symbols over a field $K$ and contacting $d$ helpers must
download at least
\begin{equation}\label{eq:intro-cutset}
 \frac{d\ell}{d-k+1}
\end{equation}
in $K$-symbol units, which is the cut-set bound. We study exact
single-erasure repair whose bandwidth attains
\eqref{eq:intro-cutset}.

Reed--Solomon (RS) codes are scalar MDS codes. If their symbol field
is $F=\F_{q^\ell}$ and $K=\F_q$, expanding each code symbol in a
$K$-basis of $F$ gives subpacketization $\ell$.
Thus the extension degree determines how finely a symbol can be
divided for repair. We require the cut-set bound to be attained
for every coordinate and every prescribed set of $d$ surviving
coordinates, and ask for the smallest possible $\ell$.
Throughout the linear part of the paper, the helper functions are
$K$-linear. Their choice may depend on both the failed coordinate
and the helper set.

\subsection{Previous Results and the Extension-Degree Problem}

The study of repair bandwidth has led to a substantial theory of MDS
array codes. Explicit constructions, including the families of
Ye and Barg~\cite{YBISIT}, show that the cut-set bound is attainable
for broad parameter ranges. For scalar codes, the field structure
places additional restrictions on the available repair maps.
Guruswami and Wootters~\cite{GW} characterized linear repair through
dual codewords. For RS codes, this gives a formulation in terms of
low-degree polynomials whose evaluations span the full symbol field
at the erased coordinate and smaller spaces at the helpers.
Ye and Barg~\cite{YBISIT} constructed RS codes with
asymptotically optimal repair bandwidth.
Dau and Milenkovic~\cite{DM} developed repair schemes for
full-length RS codes that are optimal among linear schemes
in their specified parameter regimes. Optimality among linear
schemes does not by itself imply that the cut-set bound is attained.

Tamo, Ye, and Barg~\cite{TYB} constructed RS codes attaining the
cut-set bound and established a super-exponential subpacketization
lower bound for scalar linear MDS codes. Zhao, Fu, and Fu~\cite{ZFF}
extended this lower-bound framework to a repair model in which each
failed node excludes a prescribed set of potential helpers.
The construction of Tamo, Ye, and Barg uses evaluation points
whose degrees over $K$ are distinct primes satisfying a
congruence condition. Qiu, Fang, Xia, and Fu~\cite{QFFXISIT,QFFX} subsequently
gave a basis-transformation construction of the required repair
subspaces, removing that congruence condition. In particular, their
subspace theorem applies to arbitrary integer relative degrees
at least $s=d-k+1$.

Allowing repair bandwidth above the cut-set bound can substantially
reduce subpacketization. Li, Wang, and Jafarkhani~\cite{LWJ}
constructed scalar RS codes exhibiting this tradeoff.
For a fixed code rate and any fixed $\varepsilon>0$,
Guruswami and Jiang~\cite{GJ} constructed RS codes with
subpacketization $n^{O(1/\varepsilon)}$ and repair bandwidth at most
$1+\varepsilon$ times the cut-set bound for $d=n-1$.
Chowdhury and Vardy~\cite{CV} reduced subpacketization for RS
families whose repair bandwidth approaches that bound at fixed
redundancy. For fixed composite redundancy and $d=n-1$,
Wang et al.~\cite{WLLZZ} obtained further reductions in
subpacketization while retaining asymptotically optimal bandwidth.
Their construction is formulated for rack-aware storage.
With one node per rack, repair from all surviving racks becomes
repair from all $n-1$ surviving nodes, giving the ordinary
model with $d=n-1$.

Recent work has also reduced the amount of data downloaded or read
when repairing specified RS codes.
Kim, Kruglik, and Kiah~\cite{KKK} combined the selection of helper
nodes with linear dependencies among downloaded traces to reduce
the bandwidth of single-erasure repair for full-length RS codes.
Liu, Xu, and Zhang~\cite{LXZ} studied the amount of data read at
the helpers, or I/O cost, for RS codes evaluated on subspaces.
Using exponential sums, they obtained I/O lower bounds and
matching schemes in specified two- and three-parity regimes.

Field and geometric constraints also appear in earlier repair
arguments. The scalar-MDS lower-bound proof in~\cite{TYB} uses
multiplicative invariance of repair subspaces to obtain restrictions
on generated subfields. Duursma and Dau~\cite{DD} used cross ratios
to interpret bandwidth reductions for a specific RS code.

We ask for the exact minimum extension degree when repair attains
the cut-set bound, with the base field fixed and the same
evaluation set supporting every target and every prescribed
helper set. We address this question by relating repair spaces to
the subfields generated by cross ratios of the evaluation
points. For the sufficient direction, we use the existing
construction framework together with the splitting-subspace
criterion of Aggarwal and Ram~\cite{AR} to give a short
existence proof.

Nonlinear repair is a separate issue. Con and Tamo~\cite{CT}
introduced a partition formulation of general repair schemes,
constructed nonlinear schemes that can outperform linear ones,
and obtained a stronger bandwidth lower bound over prime symbol
fields under symmetric downloads.
Our nonlinear result concerns \emph{exact equality} in the cut-set
bound. At equality, the helper partitions are coset partitions of
additive subgroups. These partitions therefore induce repair maps
that are linear over the prime field, allowing the extension-degree
lower bounds over that field to be applied to nonlinear schemes as well.

\subsection{Related Work and the First Version}

The first arXiv version of this paper was submitted on
25 October 2026~\cite{FirstVersion}.
It established the exact minimum for $n\le2d-2$, the full-helper
feasibility criterion and coordinate-count bound, and prime-field
linearization at cut-set equality.
The closely related preprint of Fang, Qiu, Xia, and Fu was submitted
on 27 October 2026~\cite{FQXF}.
The two papers obtain matching conclusions on overlapping parameter
ranges, providing additional corroboration of the common results.

For linear repair, \cite[Theorem~7]{FQXF} gives an explicit equality
range with leading term $d^2/4$; its Proposition~4 gives a stronger
condition depending on $s$. That paper also studies the arithmetic
minimum and the regime of fixed $k,d$ with $n\to\infty$.
The comparison below records the exact parameter conditions.
Our new lower bound reaches $(1-\varepsilon)d^2$ for sufficiently
large $d$. We combine it with the affine-plane upper bound in
\cite[Proposition~5]{FQXF} to locate the transition.

\subsection{Our Results}

Fix a prime power $q$ and integers $3\le k<d<n$, and put
$s=d-k+1\ge2$. Let $L_q(n,k,d)$ denote the smallest extension degree
for which some length-$n$ RS code over $\F_{q^\ell}$ admits
$\F_q$-linear repair attaining the cut-set bound for every target
and every prescribed $d$-element helper set. For $m\ge1$, define
\begin{equation}\label{eq:Q}
 Q_s(m)=\min\left\{\prod_{i=1}^{m}a_i:
 a_i\in\mathbb Z_{\ge s},\quad
 \gcd(a_i,a_j)=1\ (i\ne j)\right\},
 \qquad Q_s(0)=1.
\end{equation}
The integers in this definition need not be prime.
In the notation of~\cite{FQXF}, $G_s=Q_s$ and
$\ell_{\min}(n,k,d;q)=L_q(n,k,d)$.
Table~\ref{tab:parameter-ranges} compares the equality ranges.

\begin{table}[ht]
\centering
\small
\caption{Sufficient conditions for $L_q(n,k,d)=sQ_s(n)$.
All rows use the parameter convention above; $\pi(x)$ counts primes
at most $x$. In the last row, the threshold $d_0(\varepsilon)$
is independent of $q$ and $k$.}
\label{tab:parameter-ranges}
\begin{tabular}{p{0.25\linewidth}p{0.65\linewidth}}
\toprule
Result & Additional condition\\
\midrule
First version of this paper & $n\le2d-2$\\[3pt]
\cite[Theorem~7]{FQXF} & $d(d+2)\ge4(n+1)$\\[3pt]
\cite[Proposition~4]{FQXF}, with its construction &
$d\ge6$ and
$n+\pi(s-1)\le(\lfloor d/2\rfloor+2)\lceil d/2\rceil$\\[3pt]
Quadratic bound of this paper & $n\le(1-\varepsilon)d^2$,
$d\ge d_0(\varepsilon)$, for each fixed $0<\varepsilon<1$\\
\bottomrule
\end{tabular}
\end{table}

For $d\ge6$, the explicit ranges of~\cite{FQXF} contain the
first-version range. The latter additionally covers
$(n,k,d)=(6,3,4),(8,3,5),(8,4,5)$.
The asymptotic result holds for every base field and allows
$k$, and hence $s$, to vary with $d$.
Our conclusions are organized as follows.

\begin{enumerate}
\item \emph{A quadratic transition.}
For every $0<\varepsilon<1$, there is $d_0(\varepsilon)$ such that,
for all $d\ge d_0(\varepsilon)$,
\begin{equation}\label{eq:intro-quadratic}
 \begin{aligned}
 n\le(1-\varepsilon)d^2
 &\quad\Longrightarrow\quad L_q(n,k,d)=sQ_s(n),\\
 n\ge(1+\varepsilon)d^2
 &\quad\Longrightarrow\quad L_q(n,k,d)<sQ_s(n).
 \end{aligned}
\end{equation}
The threshold is independent of $q,k,n$.
Consequently, the first length admitting a strict improvement satisfies
$n_*(q,k,d)=(1+o(1))d^2$, uniformly in $q$ and $k$.
Thus $d^2$ is the approximate transition point.
The new part is the lower bound on the equality range and the
resulting matching asymptotic location.

\item \emph{Finite-parameter conclusions from the first version.}
For every prime power $q$,
\begin{equation}\label{eq:intro-main}
 L_q(n,k,d)=sQ_s(n)\qquad(n\le2d-2).
\end{equation}
This statement applies to small $d$ as well as large $d$.
The polynomial-graph bounds of Theorem~\ref{thm:graphs}
were also established in the first version; in particular,
\begin{equation}\label{eq:intro-smaller}
 L_q(R^2,k,d)\le sQ_s(R(d+1))<sQ_s(R^2)
 \qquad(R\ge d+2\text{ a prime power}).
\end{equation}

\item \emph{Full-helper feasibility and attainment of coordinate counts.}
For $d=n-1$, the first version characterized the degrees admitting
cut-set repair of every coordinate: $s\mid\ell$, and $\ell/s$ has
$n$ pairwise coprime divisors at least $s$.
It also bounded the number of coordinates that can attain the bound.
The coordinate-count bound is attained whenever $s\mid\ell$
and $q\ge n-2$.

\item \emph{Transfer to nonlinear repair.}
The first version proved that a nonlinear scheme at cut-set equality
induces a linear scheme over the prime field with the same
helper-output alphabet sizes.
For $q=p^e$, the new lower bound therefore applies to $e\ell$.
When the base field is prime, the nonlinear and linear minima agree
for all admissible parameters, so the same quadratic transition
holds without a linearity assumption.
\end{enumerate}

The new lower bound is proved by contradiction, starting from a
code of minimum extension degree. We compare the cross-ratio
field of its evaluation set with the fields of the sets obtained
by deleting individual points. Normalizing three points converts these degree
constraints into an arithmetic tuple. Hu and Zhang's weighted
covering inequality \cite[Lemma~2.2]{HZ} then gives a lower bound
on the product of its residual prime powers. Comparing this
product with $Q_s(n)$ rules out points whose deletion ratio is
smaller than $s$. The affine-plane construction cited above gives the
upper side of the transition.

\subsection{Organization of the Paper}

Section~\ref{sec:prelim} introduces the repair model and the
projective geometry facts used in the proofs.
Section~\ref{sec:construction} gives the existence criterion in
terms of evaluation degrees.
Section~\ref{sec:lower} derives the local field constraints,
characterizes full-helper feasibility and proves attainment of the
coordinate-count bound when $q\ge n-2$, and records the
finite-parameter lower bound of the first version.
Section~\ref{sec:quadratic} proves equality below the quadratic
scale using a finite prime-power cost inequality.
Section~\ref{sec:small} combines that result with upper bounds to
locate the first strict improvement.
Section~\ref{sec:nonlinear} proves prime-field linearization and
transfers the bounds, and Section~\ref{sec:conclusion} discusses
the remaining questions.
Appendix~\ref{app:finite} contains the proof specific to
$n\le2d-2$.

\section{Preliminaries}\label{sec:prelim}

This section introduces the notation and repair model, states the
standard cut-set and dual-code results in the form needed here,
and collects the properties of projective sublines used later.
The additional geometric lemmas are proved here so that the lower
bound can focus on the field degrees associated with repair.

\subsection{Notation and Finite Fields}

Throughout, $n,k,d,\ell$ are positive integers with
$3\le k<d<n$, $q$ is a prime power, and $s=d-k+1\ge2$.
We work over the finite fields $K=\F_q$ and $F=\F_{q^\ell}$.
All subfields lie in a fixed algebraic closure of $K$.
For a field $L$, let $L^\times=L\setminus\{0\}$ and let
$L[X]_{<u}$ consist of zero and the polynomials of degree less
than $u$. We write $\Span_E V$ for the $E$-linear span of a set
$V$. For subfields $E_1,E_2$ containing $K$, write
$E_1E_2=K(E_1\cup E_2)$ for the smallest subfield containing both.
We use the standard identities~\cite[Ch.~2]{LN}
\begin{equation}\label{eq:fields}
 \begin{split}
 \F_{q^a}\F_{q^b}&=\F_{q^{\lcm(a,b)}},\qquad
 \F_{q^a}\cap\F_{q^b}=\F_{q^{\gcd(a,b)}},\\
 \F_{q^a}\subseteq\F_{q^b}&\quad\Longleftrightarrow\quad a\mid b.
 \end{split}
\end{equation}
For a prime $\rho$ and a positive integer $u$, let $v_\rho(u)$
denote the exponent of $\rho$ in $u$, i.e.,
$v_\rho(u)=\max\{a\in\mathbb Z_{\ge0}:\rho^a\mid u\}$.
Let $p_j$ be the $j$th prime, $\Pi_j=\prod_{a=1}^j p_a$,
$\Pi_0=1$, $\pi(x)=|\{\rho\le x:\rho\text{ prime}\}|$, and
$\omega(u)=|\{\rho:\rho\text{ prime},\ \rho\mid u\}|$.
For a real number $x$, write $x_+=\max\{x,0\}$.
All logarithms without an explicit base are natural.
In addition to \eqref{eq:Q}, define, for $b\in\Zpos$,
\begin{equation}\label{eq:P}
 P_s(b)=\max\left\{m\ge0:
 \begin{array}{l}
 a_1,\ldots,a_m\mid b,\quad a_i\ge s,\\
 \gcd(a_i,a_j)=1\quad(i\ne j)
 \end{array}\right\}.
\end{equation}
The empty collection is admissible. The notation $b=\ell/s$
will be used only after $s\mid\ell$ has been established.

\subsection{Reed--Solomon Codes and Repair}

For an $n$-element set $A\subset F$, define
\[
 \mathcal C=\RS_k(A)
 =\{(f(a))_{a\in A}:f\in F[X]_{<k}\}.
\]
A repair instance is a pair $(c,H)$ with $c\in A$,
$H\subseteq A\setminus\{c\}$, and $|H|=d$.
The coordinates in $H$ are called helpers, and $H$ is the helper set.
Each helper sends fixed functions of its own symbol in one round;
helpers do not communicate with one another. The functions may
depend on $(c,H)$. Exact repair requires recovery of $f(c)$ for
every $f\in F[X]_{<k}$.
For deterministic helper maps $\mu_x:F\to\mathcal Y_x=\mu_x(F)$,
define their bandwidth in $K$-symbol units by
\begin{equation}\label{eq:repair-cost}
 B(\mu)=\sum_{x\in H}\log_q|\mathcal Y_x|.
\end{equation}
For a $K$-linear helper map,
$\log_q|\mu_x(F)|=\dim_K\mu_x(F)$.
Let $\beta_A(c;H)$ be the minimum of $B(\mu)$ over exact
$K$-linear repair schemes. A constant helper map contributes zero.
Define
\begin{equation}\label{eq:L}
 L_q(n,k,d)=\min\left\{\ell\ge1:
 \begin{array}{l}
 \exists A\subset\F_{q^\ell},\quad |A|=n,\\
 \beta_A(c;H)=d\ell/s\quad\text{for every }(c,H)
 \end{array}\right\}.
\end{equation}
The condition ``for every $(c,H)$'' in \eqref{eq:L} concerns
separate single-erasure repairs of the same code.

For a $K$-subspace $\mathcal P\subset F[X]$, write
$\mathcal P(x)=\{g(x):g\in\mathcal P\}$.
We use the dual-code characterization of Guruswami and
Wootters~\cite[Thm.~4 and Cor.~9]{GW} in the following polynomial form.

\begin{proposition}\label{prop:dual}
For $B=H\cup\{c\}$, let $\beta_B(c;H)$ denote the minimum
$K$-linear repair bandwidth for the punctured code $\RS_k(B)$.
Then
\begin{equation}\label{eq:dual}
 \beta_A(c;H)=\beta_B(c;H)
 =\min_{\substack{\mathcal P\subset F[X]_{<s}\text{ a }K\text{-subspace}\\
              \dim_K\mathcal P=\ell,\ \mathcal P(c)=F}}
       \sum_{x\in H}\dim_K\mathcal P(x).
\end{equation}
\end{proposition}

Indeed, puncturing gives $\RS_k(B)$, whose dual consists of
$(v_xg(x))_{x\in B}$ with $g\in F[X]_{<s}$ and
$v_x=\prod_{z\in B\setminus\{x\}}(x-z)^{-1}$.
These nonzero multipliers do not change the evaluation dimensions.
The restriction $\dim_K\mathcal P=\ell$ does not increase the minimum:
from any $\mathcal P$ with $\mathcal P(c)=F$, choose
$g_1,\ldots,g_\ell\in\mathcal P$ whose values at $c$ form a $K$-basis
of $F$. Then $\mathcal P_0=\Span_K\{g_1,\ldots,g_\ell\}$ satisfies
\[
 \dim_K\mathcal P_0=\ell,\qquad \mathcal P_0(c)=F,\qquad
 \mathcal P_0(x)\subseteq\mathcal P(x)\quad(x\in H).
\]

We also use the cut-set bound~\cite{DGW,TYB}, including its
equality condition for possibly unequal helper downloads:
\begin{equation}\label{eq:cutset}
 \beta_A(c;H)\ge\frac{d\ell}{s},
\end{equation}
and, for every minimizing repair space at equality,
\begin{equation}\label{eq:equal-ranks}
 s\mid\ell,\qquad
 \dim_K\mathcal P(x)=\ell/s\quad(x\in H).
\end{equation}
At equality in \eqref{eq:cutset}, condition~\eqref{eq:equal-ranks}
follows from the subset cut-set inequalities: every $s$-subset of the helpers must
contribute at least $\ell$ symbols. We will use
\eqref{eq:cutset}--\eqref{eq:equal-ranks} as standard facts.
For every deterministic exact repair scheme, the cut-set inequalities are
\begin{equation}\label{eq:alphabet-cutset}
 \sum_{x\in J}\log_q|\mathcal Y_x|\ge\ell
 \quad(J\subseteq H,\ |J|=s),\qquad
 B(\mu)\ge d\ell/s.
\end{equation}
These are the bounds used in Section~\ref{sec:nonlinear}.

\subsection{Projective Lines, Cross Ratios, and Sublines}

The lower-bound argument associates a subfield with each set of
helper evaluation points. The definitions below let us express
constraints on these subfields through containment of point sets.

Let $L/K$ be a finite extension and put
$\Pone(L)=L\cup\{\infty\}$, where $\infty$ is a formal symbol.
A projective transformation, or projectivity, is a bijection
$g:\Pone(L)\to\Pone(L)$ of the form
$g(x)=(ax+b)/(cx+d)$, where $a,b,c,d\in L$ and $ad-bc\ne0$.
For finite $x$ with $cx+d=0$, set $g(x)=\infty$; at infinity,
set $g(\infty)=a/c$ if $c\ne0$ and $g(\infty)=\infty$ if
$c=0$. These maps form a group under composition, denoted by
$\PGL_2(L)$.
For a set $C\subseteq\Pone(L)$, write $g(C)=\{g(x):x\in C\}$.
Three distinct points $u,v,w$ determine a unique projectivity
$\phi_{u,v,w}$ sending them to $0,\infty,1$, respectively
\cite[Prop.~1.8]{Cameron}. If the three reference points
$u,v,w$ are finite, then for $x\in L\setminus\{v\}$,
\begin{equation}\label{eq:phi}
 \phi_{u,v,w}(x)
 =\frac{(w-v)(x-u)}{(w-u)(x-v)}.
\end{equation}
For four pairwise distinct points, define the cross ratio by
$\chi(u,v,w,x)=\phi_{u,v,w}(x)$. Since $\phi_{u,v,w}$ is
bijective, this value belongs to $L\setminus\{0,1\}$.
For every $g\in\PGL_2(L)$,
\begin{equation}\label{eq:cross-ratio-invariance}
 \chi(g(u),g(v),g(w),g(x))=\chi(u,v,w,x)
\end{equation}
\cite[Prop.~4.20]{Cameron}. For pairwise distinct finite
$a,b,c$, the cases involving infinity are
\begin{equation}\label{eq:cross-ratio-infinity}
 \begin{aligned}
 \chi(\infty,a,b,c)&=\frac{b-a}{c-a},&
 \chi(a,\infty,b,c)&=\frac{c-a}{b-a},\\
 \chi(a,b,\infty,c)&=\frac{c-a}{c-b},&
 \chi(a,b,c,\infty)&=\frac{c-b}{c-a}.
 \end{aligned}
\end{equation}

For $C\subset\Pone(L)$ with $|C|\ge3$, define
\begin{equation}\label{eq:T}
 T(C)=K\bigl(\chi(u,v,w,x):u,v,w,x\in C
                          \text{ pairwise distinct}\bigr).
\end{equation}
If $|C|=3$, this means $T(C)=K$.
For any three distinct $u,v,w\in C$, Lemma~\ref{lem:geometry}
shows that $T(C)$ is generated over $K$ by the values
$\phi_{u,v,w}(x)$ with $x\in C\setminus\{u,v,w\}$.
For a subfield $K\subseteq E\subseteq L$, an $E$-subline is
a set
\[
 D=g(\Pone(E))=\{g(a):a\in E\}\cup\{g(\infty)\},
 \qquad g\in\PGL_2(L).
\]
For any three distinct points $u,v,w\in\Pone(L)$, the unique
$E$-subline containing them is
$D=\phi_{u,v,w}^{-1}(\Pone(E))$; see
\cite[Sec.~4.5, p.~62]{Cameron}.

\begin{lemma}\label{lem:geometry}
Let $C\subseteq\Pone(L)$ with $|C|\ge3$, and let
$u,v,w\in C$ be distinct. Then
\begin{equation}\label{eq:anchors}
 T(C)=K\bigl(\phi_{u,v,w}(x):x\in C\setminus\{u,v,w\}\bigr).
\end{equation}
For every $g\in\PGL_2(L)$, we have $T(g(C))=T(C)$.
For every subfield $K\subseteq E\subseteq L$,
\begin{equation}\label{eq:subline-field}
 T(C)\subseteq E
 \quad\Longleftrightarrow\quad
 C\subseteq D\text{ for an }E\text{-subline }D.
\end{equation}
For $i\notin C$ and for distinct $i,j\in C$, respectively,
\begin{align}
 T(C\cup\{i\})&=T(C)(\phi_{u,v,w}(i)),
                                      \label{eq:add-point}\\
 T(C)&=T(C\setminus\{i\})T(C\setminus\{j\})
       \qquad(|C|\ge5).                \label{eq:delete}
\end{align}
\end{lemma}

\begin{proof}
Let $E_0$ be the field on the right of \eqref{eq:anchors}.
Each of its displayed generators is a cross ratio of points of
$C$, so $E_0\subseteq T(C)$. Conversely,
$\phi_{u,v,w}(C)\subseteq\Pone(E_0)$, and every cross ratio
of four points in $\Pone(E_0)$ belongs to $E_0$ by
\eqref{eq:phi} and \eqref{eq:cross-ratio-infinity}.
Equation~\eqref{eq:cross-ratio-invariance} gives $T(C)\subseteq E_0$.
It also gives $T(g(C))=T(C)$ for every $g\in\PGL_2(L)$.
If $T(C)\subseteq E$, then
$C\subseteq\phi_{u,v,w}^{-1}(\Pone(E))$.
Conversely, if $C\subseteq D=\psi(\Pone(E))$, then
$\psi^{-1}(C)\subseteq\Pone(E)$, so
$T(C)=T(\psi^{-1}(C))\subseteq E$.
Applying \eqref{eq:anchors} to $C\cup\{i\}$ with the same
reference points $u,v,w$ gives \eqref{eq:add-point}.
To prove \eqref{eq:delete}, choose
three reference points $u,v,w\in C\setminus\{i,j\}$. Using
these same points for both fields gives
\[
 T(C\setminus\{i\})T(C\setminus\{j\})
 =K\bigl(\phi_{u,v,w}(x):x\in C\setminus\{u,v,w\}\bigr)
 =T(C).
\]
\end{proof}

\begin{lemma}\label{lem:subline-intersection}
Let $E\subseteq E'\subseteq L$. If an $E$-subline $D$ and
an $E'$-subline $D'$ have three common points, then $D\subseteq D'$.
In particular, two distinct $E$-sublines have at most two common
points.
\end{lemma}

\begin{proof}
Let $u,v,w$ be three common points and set
$\theta=\phi_{u,v,w}$, so that
$\theta(u)=0$, $\theta(v)=\infty$, and $\theta(w)=1$.
Uniqueness of a subline through $0,\infty,1$ gives
$\theta(D)=\Pone(E)$ and $\theta(D')=\Pone(E')$.
Since $E\subseteq E'$, we have $\theta(D)\subseteq\theta(D')$
and hence $D\subseteq D'$. Taking $E=E'$ proves the last
assertion.
\end{proof}

The next lemma is needed only for the parameter case
$(n,k,d)=(6,3,4)$, treated in Lemma~\ref{lem:six}.

\begin{lemma}\label{lem:six-geometry}
Suppose $[L:E]$ is an odd prime. There do not exist a set
$C\subseteq\Pone(L)$ with $|C|=6$ and three distinct
$E$-sublines $D_1,D_2,D_3$ such that $|C\cap D_i|=4$ for
$i=1,2,3$ and the two-point sets $C\setminus D_i$ are pairwise
disjoint.
\end{lemma}

\begin{proof}
The sets $C\cap D_i$ have pairwise intersections of size two.
Let $p,q$ be the two points in $C\cap D_1\cap D_2$ and
choose $a\in(C\cap D_1)\setminus D_2$. Set
$\theta=\phi_{p,q,a}$, so that $\theta(p)=0$,
$\theta(q)=\infty$, and $\theta(a)=1$. Set
\[
 \widetilde C=\theta(C),\qquad
 \widetilde D_i=\theta(D_i)\quad(i=1,2,3).
\]
The map $\theta$ is bijective, so the intersection sizes and
disjointness assumptions are preserved. Since
$0,\infty,1\in\widetilde D_1$, uniqueness of the $E$-subline
through these three points gives $\widetilde D_1=\Pone(E)$.
Write
\[
 \widetilde C\cap\widetilde D_1=\{0,\infty,1,u\},\qquad
 \widetilde C\cap\widetilde D_2=\{0,\infty,x,vx\},
 \qquad u,v\in E\setminus\{0,1\}.
\]
Indeed, for a finite nonzero point
$x\in\widetilde C\cap\widetilde D_2$, the
unique $E$-subline containing $0,\infty,x$ is
$\widetilde D_2=\{xe:e\in E\}\cup\{\infty\}$.
The three pairwise disjoint two-element sets
$\widetilde C\setminus\widetilde D_j$ partition $\widetilde C$.
The first two are $\{x,vx\}$ and $\{1,u\}$, so
$\widetilde C\setminus\widetilde D_3=\{0,\infty\}$.
Hence $\widetilde C\cap\widetilde D_3=\{1,u,x,vx\}$,
whose cross ratios lie in $E$.
In particular,
\[
 r=\chi(x,vx,u,1)
  =\frac{(1-x)(u-vx)}{(1-vx)(u-x)}\in E\setminus\{0,1\}.
\]
The four points $1,u,x,vx$ are distinct, so
$(1-vx)(u-x)\ne0$. Clearing denominators and expanding gives
\begin{equation}\label{eq:six-quadratic}
 (1-r)vx^2+[-u-v+r(1+uv)]x+(1-r)u=0.
\end{equation}
Its leading coefficient $(1-r)v$ is nonzero, since $r\ne1$
and $v\ne0$. Hence $[E(x):E]\le2$.
Since this degree divides the odd prime $[L:E]$, we obtain
$x\in E$. This gives
$\widetilde D_1=\widetilde D_2=\Pone(E)$ and hence
$D_1=D_2$, a contradiction.
\end{proof}

\section{A Simple Existence Proof for RS Codes Attaining the Cut-Set Bound}
\label{sec:construction}

This section applies a splitting-subspace criterion to give a short
existence proof and formulate a sufficient condition on evaluation
degrees. The same condition will be used in Section~\ref{sec:small}.

\subsection{Splitting Subspaces}

The following specialization of \cite[Cor.~2.4]{AR} is obtained
by applying that result to the multiplication map
$M_\alpha:L\to L$, $x\mapsto\alpha x$.

\begin{theorem}[Aggarwal--Ram {\cite{AR}}]\label{thm:split}
Let $E\subseteq L$ be finite fields, let $N=[L:E]$, and suppose
$s\mid N$. For $\alpha\in L$, there exists an $E$-subspace
$U\subseteq L$ satisfying
\begin{equation}\label{eq:split}
 \dim_EU=N/s,\qquad
 L=\bigoplus_{j=0}^{s-1}\alpha^jU
\end{equation}
if and only if $[E(\alpha):E]\ge s$.
\end{theorem}

In the repair construction, we apply this theorem with $L=F$,
$E=K(x:x\in H)$, and $\alpha=c$, where $c$ is the target
and $H$ is the helper set. For $\mathcal P=U[X]_{<s}$, the
space of polynomials of degree less than $s$ with coefficients
in $U$, the splitting condition gives $\mathcal P(c)=F$, and
$\mathcal P(x)=U$ for every $x\in H$. Since $\dim_KU=\ell/s$,
Proposition~\ref{prop:dual} gives repair cost $d\ell/s$,
attaining the cut-set bound \eqref{eq:cutset}.

\subsection{A Sufficient Condition on Evaluation Degrees}

To apply Theorem~\ref{thm:split} to every helper set, we require
the target point to have relative degree at least $s$ over the
field generated by those helpers. The quotient in
\eqref{eq:degree-condition} expresses this relative degree using
the degrees of the evaluation points over $K$.

For the later arithmetic argument, let $m>D\ge1$ and $s\ge2$.
A tuple of positive integers $(a_1,\ldots,a_m)$ is
\emph{$(D,s)$-feasible} if
\begin{equation}\label{eq:tuple-feasible}
 \frac{\lcm(a_i,a_j:j\in J)}{\lcm(a_j:j\in J)}\ge s
 \quad(i\notin J\subseteq\{1,\ldots,m\},\ |J|=D).
\end{equation}
Write $\Delta_s(m,D)$ for the minimum of $\lcm(a_1,\ldots,a_m)$
over such tuples. This agrees with the arithmetic quantity
in~\cite{FQXF}.

\begin{proposition}\label{prop:exist}
Let $b\in\Zpos$, and let $f_1,\ldots,f_n$ be distinct positive
divisors of $b$ such that
\begin{equation}\label{eq:degree-condition}
 \frac{f_i}{\gcd(f_i,\lcm_{j\in J}f_j)}\ge s
 \quad\text{for every }i\notin J\subseteq\{1,\ldots,n\},\ |J|=d.
\end{equation}
Then there exists a length-$n$ RS code over $\F_{q^{sb}}$
that attains the cut-set bound for every repair instance.
\end{proposition}

\begin{proof}
Set $F=\F_{q^{sb}}$ and choose $\alpha_i\in\F_{q^{f_i}}$
of exact degree $f_i$ over $K$. These elements can be chosen as
roots of irreducible polynomials of the respective degrees over $K$.
Distinct degrees give distinct points, so
$A=\{\alpha_1,\ldots,\alpha_n\}$ has size $n$.
Fix $i\notin J$, $|J|=d$, and put
\[
 u_J=\lcm_{j\in J}f_j,\qquad
 E_J=K(\alpha_j:j\in J)=\F_{q^{u_J}}.
\]
The field identities \eqref{eq:fields} give
\[
 [E_J(\alpha_i):E_J]
 =\frac{f_i}{\gcd(f_i,u_J)}\ge s,
 \qquad [F:E_J]=s\frac{b}{u_J}.
\]
Theorem~\ref{thm:split} therefore supplies an $E_J$-space
$U\subseteq F$ such that
\[
 \dim_{E_J}U=b/u_J,\qquad \dim_KU=b,
 \qquad F=\bigoplus_{a=0}^{s-1}\alpha_i^aU.
\]
For this repair instance take
\begin{equation}\label{eq:repair-space-construction}
 \mathcal P=U[X]_{<s}
 =\left\{\sum_{a=0}^{s-1}u_aX^a:u_a\in U\right\}.
\end{equation}
Then
\[
 \dim_K\mathcal P=sb,\qquad
 \mathcal P(\alpha_i)=F,\qquad
 \mathcal P(\alpha_j)=\sum_{a=0}^{s-1}\alpha_j^aU=U
 \quad(j\in J).
\]
The last equality uses $\alpha_j\in E_J$ and the constant
term in \eqref{eq:repair-space-construction}.
Proposition~\ref{prop:dual} gives bandwidth $db=d\ell/s$.
The points $A$ are fixed, although $U$ may depend on $(i,J)$.
\end{proof}

\begin{corollary}\label{cor:tuple-realization}
Under the standing coding parameters,
\begin{equation}\label{eq:tuple-upper}
 L_q(n,k,d)\le s\Delta_s(n,d)\le sQ_s(n).
\end{equation}
\end{corollary}
\begin{proof}
Let $(a_1,\ldots,a_n)$ be $(d,s)$-feasible and put
$b_0=\lcm(a_1,\ldots,a_n)$.
The $a_i$ are distinct: if $a_i=a_j$ with $i\ne j$, a helper
set containing $j$ and excluding $i$ would give ratio $1$ in
\eqref{eq:tuple-feasible}, contrary to $s\ge2$.
They are divisors of $b_0$ satisfying
\eqref{eq:degree-condition}, so Proposition~\ref{prop:exist}
gives a code of extension degree $sb_0$.
Minimizing proves the first inequality.
Any pairwise coprime $n$-tuple with entries at least $s$ is
feasible and has least common multiple equal to its product,
which proves the second.
\end{proof}

\subsection{Attaining the General Upper Bound}

\begin{corollary}\label{cor:upper}
For every prime power $q$ and all $3\le k<d<n$,
\begin{equation}\label{eq:upper}
 L_q(n,k,d)\le sQ_s(n).
\end{equation}
More generally, if $s\mid\ell$ and $P_s(\ell/s)\ge n$,
there is a length-$n$ RS code over $\F_{q^\ell}$ admitting
$K$-linear repair at the cut-set bound for every target and
every $d$-element helper set.
\end{corollary}

\begin{proof}
Put $b=\ell/s$, and choose pairwise coprime divisors
$a_1,\ldots,a_n$ of $b$ with $a_i\ge s$.
For $f_i=a_i$ and any $J$ excluding $i$,
\[
 \gcd(f_i,\lcm_{j\in J}f_j)=1.
\]
The degrees are distinct because they are pairwise coprime and
each is at least $2$. Proposition~\ref{prop:exist} applies.
Taking a minimizing tuple in \eqref{eq:Q} gives the first claim.
\end{proof}

\subsection{Coprime-Product Estimates}

We use the following extension and strict-growth estimates for $Q_s$.
\begin{lemma}\label{lem:coprime-cost}
Let $c_s=\pi(s-1)$. For $0\le z\le n$,
\begin{equation}\label{eq:Qgrowth}
 Q_s(n)\le Q_s(n-z)p_{n+c_s}^{z},\qquad
 Q_s(n)\le\frac{\Pi_{n+c_s}}{\Pi_{c_s}}.
\end{equation}
For $m\ge1$,
\begin{equation}\label{eq:Qstrict}
 Q_s(m)\ge sQ_s(m-1)>Q_s(m-1).
\end{equation}
\end{lemma}
\begin{proof}

For $m\ge0$, choose a minimizing $m$-tuple for $Q_s(m)$,
using the empty tuple when $m=0$. A factor divisible
by a prime $\rho\ge s$ must equal $\rho$: otherwise replacing
it by $\rho$ would preserve admissibility and strictly decrease
the product, contradicting minimality.
List the prime divisors at least $s$ that occur as
$\rho_1<\cdots<\rho_r$. Each is an entire factor of the tuple,
so $r\le m$. If a smaller prime at least $s$ were missing,
replacing an occurring prime by that unused prime would decrease
the product while preserving admissibility. Hence
$\rho_j=p_{c_s+j}$ for $1\le j\le r$.
Every other prime divisor is below $s$, so its index is at most
$c_s$. Thus all prime divisors of the tuple
are among the first $m+c_s$ primes. Adjoining $p_{m+c_s+1}$ gives
\[
 Q_s(m+1)\le Q_s(m)p_{m+c_s+1}.
\]
Iteration proves the first inequality in \eqref{eq:Qgrowth}.
The first $n$ primes at least $s$ give the second.

For \eqref{eq:Qstrict}, remove one factor from a minimizing
$m$-tuple. This factor is at least $s$, and the product of the
remaining factors is at least $Q_s(m-1)$.
\end{proof}

\section{Field Constraints and Finite-Parameter Results}\label{sec:lower}

This section derives the field-degree conditions forced by cut-set
equality. It also characterizes full-helper feasibility and proves
that the coordinate-count bound is attained when $q\ge n-2$.
The deletion-field identities at the end of the section provide
the common starting point for the finite and quadratic lower bounds.

\subsection{Field Constraints from Cut-Set Equality}

\begin{lemma}\label{lem:field}
If $(c,H)$ attains the cut-set bound, then $s\mid\ell$ and,
with $b=\ell/s$,
\begin{equation}\label{eq:local}
 [T(H):K]\mid b,
 \qquad [T(H\cup\{c\}):T(H)]\ge s.
\end{equation}
\end{lemma}

\begin{proof}
We prove $[T(H):K]\mid b$ by viewing a helper evaluation space
as a vector space over a field containing $T(H)$.
The full evaluation space at the target will then give the
relative-degree inequality.
Choose a minimizing repair space $\mathcal P$ in
\eqref{eq:dual}. By \eqref{eq:equal-ranks}, $s\mid\ell$ and
$\dim_K\mathcal P(x)=b$ for every $x\in H$.
Choose distinct $z,w,y_0\in H$ and a set
$Y\subseteq H\setminus\{z,w\}$ of size $s$ containing $y_0$.
This is possible because $d-2\ge s$, using $k\ge3$.
For $y\in Y$ define the Lagrange polynomial
\[
 L_y(X)=\prod_{u\in Y\setminus\{y\}}\frac{X-u}{y-u}.
\]
Consider the $K$-linear evaluation map
\[
 \operatorname{ev}_Y:\mathcal P\longrightarrow
 \prod_{y\in Y}\mathcal P(y),\qquad
 f\longmapsto(f(y))_{y\in Y}.
\]
It is injective because evaluation on $Y$ is injective on
$F[X]_{<s}$. Moreover,
\[
 \dim_K\mathcal P=sb
 =\sum_{y\in Y}\dim_K\mathcal P(y).
\]
Thus $\operatorname{ev}_Y$ is an isomorphism. Its inverse sends
$(u_y)_{y\in Y}$ to $\sum_{y\in Y}L_y(X)u_y$, so
\begin{equation}\label{eq:lagrange-space}
 \mathcal P=\bigoplus_{y\in Y}L_y(X)\mathcal P(y).
\end{equation}
At $z$, every summand $L_y(z)\mathcal P(y)$ has dimension $b$
and their sum also has dimension $b$. Thus
\[
 L_y(z)\mathcal P(y)=\mathcal P(z)=:U\ne\{0\}
 \quad(y\in Y).
\]
Set
\begin{equation}\label{eq:stabilizer-field}
 E=\{\lambda\in F:\lambda U\subseteq U\}.
\end{equation}
This set contains $K$ and is closed under addition, subtraction,
and multiplication. For $0\ne\lambda\in E$, multiplication
by $\lambda$ is injective on the finite set $U$, so
$\lambda U=U$ and $\lambda^{-1}U=U$.
Therefore $E$ is a field, $U$ is an $E$-space, and
\begin{equation}\label{eq:stabilizer-dimension}
 b=[E:K]\dim_EU.
\end{equation}

We next show that $T(H)\subseteq E$.
For $y\in Y$ and $x\in F\setminus Y$, put
$D_y(x)=L_y(x)/L_y(z)$. For $y\in Y\setminus\{y_0\}$ and
$x\in F\setminus Y$, define
\begin{equation}\label{eq:R}
 R_y(x):=\frac{D_y(x)}{D_{y_0}(x)}
 =\frac{(z-y)(x-y_0)}{(z-y_0)(x-y)}.
\end{equation}
If $x\in H\setminus Y$, \eqref{eq:lagrange-space} gives
$\mathcal P(x)=\sum_{y\in Y}D_y(x)U$.
All summands and their sum have dimension $b$, so
$R_y(x)U=U$ and hence $R_y(x)\in E^\times$.
Let $\phi=\phi_{z,w,y_0}$, so that
$\phi(z)=0$, $\phi(w)=\infty$, and $\phi(y_0)=1$, and put
$a_y=\phi(y)$ for $y\in Y\setminus\{y_0\}$.
Substitution in \eqref{eq:R} yields
\begin{equation}\label{eq:ay}
 a_y=R_y(w)\in E\setminus\{0,1\}.
\end{equation}
For $x\in F\setminus(Y\cup\{w\})$, it also gives
\begin{equation}\label{eq:inverse-R}
 R_y(x)=\frac{a_y(\phi(x)-1)}{\phi(x)-a_y},
 \qquad \phi(x)=\frac{a_y(R_y(x)-1)}{R_y(x)-a_y}.
\end{equation}
For $x\in H\setminus(Y\cup\{z,w\})$, the denominator
$R_y(x)-a_y$ cannot vanish: the first identity in
\eqref{eq:inverse-R} would then give $a_y(a_y-1)=0$.
As $s\ge2$, one may fix $y\ne y_0$ in $Y$.
Together with \eqref{eq:ay} and
$\phi(z)=0$, $\phi(w)=\infty$, $\phi(y_0)=1$, this proves
\[
 \phi(H)\subseteq\Pone(E),\qquad T(H)\subseteq E.
\]
Equation~\eqref{eq:stabilizer-dimension} gives $[T(H):K]\mid b$.

For the relative-degree bound, put $\mu=\phi(c)$ and set
$R_{y_0}(c)=1$.
Equations~\eqref{eq:R}--\eqref{eq:inverse-R} show that
$R_y(c)\in E(\mu)$ for every $y\in Y$. Hence
\[
 F=\mathcal P(c)
  =D_{y_0}(c)\sum_{y\in Y}R_y(c)U
  \subseteq D_{y_0}(c)\Span_{E(\mu)}U.
\]
If $h=[E(\mu):E]$ and $\eta_1,\ldots,\eta_h$ is an
$E$-basis of $E(\mu)$, then
\[
 \Span_{E(\mu)}U=\sum_{a=1}^h\eta_aU,
 \qquad sb=\ell\le hb.
\]
Thus $h\ge s$. By \eqref{eq:add-point},
$T(H\cup\{c\})=T(H)(\mu)$. Since $T(H)\subseteq E$, the
minimal polynomial of $\mu$ over $E$ divides its minimal
polynomial over $T(H)$, viewed as a polynomial in $E[X]$.
Comparing their degrees gives
\[
 [T(H\cup\{c\}):T(H)]
 \ge [E(\mu):E]\ge s.
\]
\end{proof}

\subsection{All Available Helpers: Coordinate Counts and Exact Feasibility}

In this subsection only, let $d=n-1$ and $s=n-k$.
We first bound the number of coordinates that attain the cut-set
bound in a fixed code and show that this bound is attained when
$q\ge n-2$. Requiring all $n$
coordinates to attain the cut-set bound gives a necessary condition
that matches the sufficient condition in Corollary~\ref{cor:upper}.
Define
\[
 \mathcal M(A)=\{i\in A:
 \beta_A(i;A\setminus\{i\})=(n-1)\ell/s\},
 \qquad N_{\rm cs}(A)=|\mathcal M(A)|.
\]

\begin{theorem}
\label{thm:count}
Let $q$ be a prime power, and let $\ell,k,n$ be integers with
$\ell\ge1$, $k\ge3$, and $n-k\ge2$.
Set $d=n-1$, $s=n-k$, $K=\F_q$, and $F=\F_{q^\ell}$.
For every $n$-element evaluation set $A\subset F$, the number
$N_{\rm cs}(A)$ of coordinates admitting $K$-linear repair at
the cut-set bound from all $n-1$ surviving coordinates satisfies
\begin{equation}\label{eq:count}
 N_{\rm cs}(A)\le
 \begin{cases}
 0,&s\nmid\ell,\\[2pt]
 \min\{n,\max\{1,P_s(\ell/s)\}\},&s\mid\ell.
 \end{cases}
\end{equation}
If $s\mid\ell$ and $q\ge n-2$, there exists an $n$-element
evaluation set $A\subset F$ attaining equality in \eqref{eq:count}.
\end{theorem}

\begin{proof}
The first case follows from \eqref{eq:equal-ranks}.
Suppose $s\mid\ell$ and put $b=\ell/s$.
If $|\mathcal M(A)|\le1$, the claimed bound is immediate.
Assume $|\mathcal M(A)|\ge2$ and write
\[
 T=T(A),\quad t=[T:K],\quad
 h_i=[T(A\setminus\{i\}):K].
\]
Lemma~\ref{lem:field} gives $h_i\mid b$ and $t/h_i\ge s$
for $i\in\mathcal M(A)$.
Since $n=k+s\ge5$, \eqref{eq:delete} and \eqref{eq:fields} give
\begin{equation}\label{eq:lcm-deletion}
 t=\lcm(h_i,h_j)\mid b
 \quad(i,j\in\mathcal M(A),\ i\ne j).
\end{equation}
Therefore the integers $a_i=t/h_i$ satisfy
\[
 a_i\mid b,\qquad a_i\ge s,
 \qquad \gcd(a_i,a_j)
 =\frac{t}{\lcm(h_i,h_j)}=1.
\]
They form an admissible family in \eqref{eq:P}, so
$|\mathcal M(A)|\le P_s(b)$. Together with
$|\mathcal M(A)|\le n$, this proves the bound.

We now prove attainment when $q\ge n-2$. Put
$m=\min\{n,\max\{1,P_s(b)\}\}$.
For $m\ge2$ we use pairwise coprime evaluation degrees; for
$m=1$ we construct the target and helper points explicitly.
Suppose first that $m\ge2$. Choose pairwise coprime divisors
$a_1,\ldots,a_m\mid b$ with $a_i\ge s$, and choose
$\alpha_i\in F$ of exact degree $a_i$ over $K$.
The $\alpha_i$ are distinct and lie outside $K$.
Since $n-m\le n-2\le q$, choose $B\subseteq K$ of size $n-m$
and set $A=B\cup\{\alpha_1,\ldots,\alpha_m\}$.
For each $i$, put
\[
 b_i=\prod_{j\ne i}a_j,\qquad E_i=\F_{q^{b_i}}.
\]
The product $\prod_{j=1}^m a_j$ divides $b$, so $E_i\subseteq F$.
Every helper in $A\setminus\{\alpha_i\}$ belongs to $E_i$, and
\[
 [E_i(\alpha_i):E_i]=a_i\ge s,
 \qquad [F:E_i]=s\frac{b}{b_i}.
\]
Theorem~\ref{thm:split} gives an $E_i$-subspace $U_i\subseteq F$
such that $\dim_KU_i=b$ and
$F=\bigoplus_{j=0}^{s-1}\alpha_i^jU_i$.
Set $\mathcal P_i=U_i[X]_{<s}$. Since
$A\setminus\{\alpha_i\}\subseteq E_i$, we have
\[
 \begin{aligned}
 \dim_K\mathcal P_i&=s\dim_KU_i=sb=\ell,\\
 \mathcal P_i(\alpha_i)&=\sum_{j=0}^{s-1}\alpha_i^jU_i=F,\\
 \mathcal P_i(x)&=\sum_{j=0}^{s-1}x^jU_i=U_i
       \qquad(x\in A\setminus\{\alpha_i\}).
 \end{aligned}
\]
Proposition~\ref{prop:dual} and the cut-set bound give
\[
 \beta_A(\alpha_i;A\setminus\{\alpha_i\})=(n-1)b
       \quad(1\le i\le m),\qquad N_{\rm cs}(A)\ge m.
\]

It remains to treat $m=1$. Set $E=\F_{q^b}$, choose
$\mu\in F$ of degree $s$ over $E$, and choose
$\gamma\in F\setminus(E\cup\{\mu\})$.
The latter choice is possible because
$|F|=|E|^s\ge |E|+2$.
Define the finite set
\[
 H_0=\{0\}\cup\{(a-\gamma)^{-1}:a\in E\},
 \qquad c=(\mu-\gamma)^{-1}.
\]
Then $c\notin H_0$ and $|H_0|=q^b+1\ge q+1\ge n-1$.
Choose $H\subseteq H_0$ of size $n-1$ and set $A=H\cup\{c\}$.
Use the polynomial space
\[
 \mathcal Q=\left\{
 \sum_{j=0}^{s-1}u_jX^{s-1-j}(1+\gamma X)^j:
 u_0,\ldots,u_{s-1}\in E\right\}\subseteq F[X]_{<s}.
\]
The identity $1+\gamma x=ax$ at $x=(a-\gamma)^{-1}$ gives
\[
 \begin{aligned}
 \mathcal Q(c)&=c^{s-1}\sum_{j=0}^{s-1}\mu^jE=F,\qquad
 \mathcal Q(0)=E,\\
 \mathcal Q((a-\gamma)^{-1})&=(a-\gamma)^{1-s}E
 \qquad(a\in E).
 \end{aligned}
\]
Consequently,
\[
 \begin{aligned}
 \ell=\dim_K\mathcal Q(c)
 &\le\dim_K\mathcal Q\le s\dim_KE=sb=\ell,\\
 \sum_{x\in H}\dim_K\mathcal Q(x)
 &=(n-1)b=\frac{(n-1)\ell}{s}.
 \end{aligned}
\]
Proposition~\ref{prop:dual} and the cut-set bound therefore give
$\beta_A(c;H)=(n-1)\ell/s$.
In both cases, the upper bound already proved yields
$N_{\rm cs}(A)=m$, as required.
\end{proof}

\subsubsection{Feasibility at Every Coordinate}

The following criterion characterizes repair of all coordinates
at the cut-set bound.

\begin{theorem}\label{thm:all}
Let $k,n$ be integers with $k\ge3$ and $n-k\ge2$,
and set $s=n-k$ and $d=n-1$.
For every prime power $q$ and every $\ell\in\Zpos$,
\begin{equation}\label{eq:all-criterion}
 \begin{split}
 &\exists A\subseteq\F_{q^\ell},\ |A|=n,
 \quad \beta_A(i;A\setminus\{i\})=\frac{(n-1)\ell}{s}
       \quad\text{for every }i\in A\\
 &\hspace{35mm}\Longleftrightarrow\quad
 s\mid\ell\ \text{ and }\ P_s(\ell/s)\ge n.
 \end{split}
\end{equation}
In particular,
\begin{equation}\label{eq:all-minimum}
 L_q(n,k,n-1)=sQ_s(n).
\end{equation}
\end{theorem}

\begin{proof}
Necessity in \eqref{eq:all-criterion} is Theorem~\ref{thm:count},
and sufficiency is Corollary~\ref{cor:upper}.
If the conditions in \eqref{eq:all-criterion} hold, there are
$n$ pairwise coprime divisors
$a_1,\ldots,a_n\ge s$ of $b=\ell/s$ satisfying
\[
 Q_s(n)\le\prod_{i=1}^n a_i,\qquad \prod_{i=1}^n a_i\mid b.
\]
Hence $\ell\ge sQ_s(n)$. Corollary~\ref{cor:upper} attains equality.
\end{proof}

\begin{proposition}\label{prop:prime-packing}
Write $b=\prod_{j=1}^u\rho_j^{e_j}$ with distinct primes
$\rho_j$. Then $P_s(b)$ is the maximum number of pairwise
disjoint nonempty index sets $I_1,\ldots,I_m\subseteq\{1,\ldots,u\}$
such that
\begin{equation}\label{eq:prime-packing}
 \prod_{j\in I_a}\rho_j^{e_j}\ge s\quad(1\le a\le m).
\end{equation}
Consequently $P_s(b)\le\omega(b)$ and $P_2(b)=\omega(b)$.
For $s=2$, \eqref{eq:all-criterion} becomes
\[
 2\mid\ell\quad\text{and}\quad\omega(\ell/2)\ge n,
\]
and $L_q(n,n-2,n-1)=2\Pi_n$.
\end{proposition}

\begin{proof}
Given pairwise coprime divisors $a_1,\ldots,a_m\ge s$ of $b$,
put $I_h=\{j:\rho_j\mid a_h\}$ for $1\le h\le m$.
Then
\[
 I_h\cap I_{h'}=\varnothing\quad(h\ne h'),\qquad
 a_h\mid\prod_{j\in I_h}\rho_j^{e_j},\qquad
 \prod_{j\in I_h}\rho_j^{e_j}\ge a_h\ge s.
\]
Conversely, the products in \eqref{eq:prime-packing} are
pairwise coprime divisors of $b$. There can be at most $u$
disjoint nonempty subsets of $\{1,\ldots,u\}$; when $s=2$,
the $u$ sets $\{1\},\ldots,\{u\}$ all satisfy
\eqref{eq:prime-packing}.
Finally, choosing one prime divisor from each factor in
\eqref{eq:Q} gives $Q_2(n)\ge\Pi_n$, while the first $n$
primes attain equality.
\end{proof}

\begin{remark}\label{rem:singleton}
The term $\max\{1,P_s(b)\}$ in \eqref{eq:count} is necessary:
when $b=1$, we have $P_s(1)=0$, but the attainment statement in
Theorem~\ref{thm:count} gives a code with one coordinate attaining
the cut-set bound whenever $q\ge n-2$.

If $s\mid\ell$, Theorem~\ref{thm:count} also implies that
at least $n-\min\{n,\max\{1,P_s(\ell/s)\}\}$ coordinates $i$ satisfy
\[
 \beta_A(i;A\setminus\{i\})\ge (n-1)\ell/s+1,
\]
because the repair cost of a $K$-linear scheme is an integer.
\end{remark}

\subsection{Arbitrary Helper Sets}

We return to $3\le k<d<n$ and assume that the RS code with
evaluation set $A$ admits $K$-linear repair at the cut-set
bound for every target and every $d$-element helper set,
as required in \eqref{eq:L}.
For every $H\subset A$ with $|H|=d<n$, choose a target
$c\in A\setminus H$. Lemma~\ref{lem:field} gives
$T(H)\subseteq\F_{q^b}$, where $b=\ell/s$.
Since $d\ge4$, each four-point cross ratio of $A$ occurs in
some such $H$. Hence, with $t=[T(A):K]$,
\begin{equation}\label{eq:intrinsic}
 T(A)\subseteq\F_{q^b},\qquad t\mid b,\qquad \ell\ge st.
\end{equation}
It remains to obtain a lower bound on $t$.
The arguments below apply for general $n>d$.
We translate the relative-degree constraints into subline
containments and then separate the deletion factors $S_i$ from
the remaining factor $N$. The set $Z$ will record the points
with $S_i<s$.
In Appendix~\ref{app:finite}, the restriction $n\le2d-2$
ensures the three-point overlap used in
Lemma~\ref{lem:fixed-target} and yields the arithmetic
index bound \eqref{eq:index-bound}.

For the geometric arguments below, set
\[
 L=T(A)=\F_{q^t}.
\]
Choose distinct $u,v,w\in A$, put
$\theta=\phi_{u,v,w}\in\PGL_2(F)$, and set $A'=\theta(A)$.
By \eqref{eq:anchors}, every point of
$A'\setminus\{0,\infty,1\}$ belongs to $T(A)=L$.
Together with the invariance of $T(C)$ under $\theta$, this gives
\[
 A'\subseteq\Pone(L),\qquad
 T(\theta(C))=T(C)\quad(C\subseteq A,\ |C|\ge3).
\]
In particular, for every original repair instance $(c,H)$,
\[
 [T(\theta(H)\cup\{\theta(c)\}):T(\theta(H))]
 =[T(H\cup\{c\}):T(H)]\ge s.
\]
For the geometric argument, identify each original point $a\in A$
with its image $\theta(a)\in A'$ and write $A$ for $A'$.
The RS code continues to use the original evaluation set.

To compare the relative-degree constraints from different helper
sets, we track each prime exponent separately. The sublines below
translate an increase in such an exponent into a containment
condition: the helpers lie on the subline, while the target lies
outside it.

For each prime $\rho\mid t$, put $a_\rho=v_\rho(t)$ and define
\begin{equation}\label{eq:levels}
 E_{\rho,j}=\F_{q^{t/\rho^{a_\rho-j+1}}},
 \qquad 1\le j\le a_\rho.
\end{equation}
For each fixed $\rho$, they form the chain
\[
 E_{\rho,1}\subseteq E_{\rho,2}\subseteq\cdots
 \subseteq E_{\rho,a_\rho}.
\]
For $C\subseteq A$, $|C|\ge3$, write $t_C=[T(C):K]$.
For every prime $\sigma$,
\[
 v_\sigma([E_{\rho,j}:K])=
 \begin{cases}
  j-1,&\sigma=\rho,\\
  v_\sigma(t),&\sigma\ne\rho.
 \end{cases}
\]
Since $t_C\mid t$, it follows that
$t_C\mid[E_{\rho,j}:K]$ if and only if $v_\rho(t_C)<j$.
The finite-field containment criterion in \eqref{eq:fields}
and Lemma~\ref{lem:geometry} therefore give
\begin{equation}\label{eq:threshold}
 v_\rho(t_C)<j
 \quad\Longleftrightarrow\quad T(C)\subseteq E_{\rho,j}
 \quad\Longleftrightarrow\quad
 C\subseteq D\text{ for an }E_{\rho,j}\text{-subline }D.
\end{equation}
For a helper set $H$ with $T(H)\subseteq E_{\rho,j}$,
choose any three distinct points $u_0,v_0,w_0\in H$ and define
\[
 D_{\rho,j}(H)
 =\phi_{u_0,v_0,w_0}^{-1}(\Pone(E_{\rho,j})).
\]
Equation~\eqref{eq:threshold} shows that this subline contains
$H$. Since three distinct points determine a unique
$E_{\rho,j}$-subline, the definition is independent of the
chosen three points. If $T(H)\not\subseteq E_{\rho,j}$, no
$E_{\rho,j}$-subline contains $H$, and $D_{\rho,j}(H)$ is
undefined.
For $i\notin H$, uniqueness implies
\begin{equation}\label{eq:witness}
 \begin{split}
 v_\rho(t_H)<j\le v_\rho(t_{H\cup\{i\}})
 \quad\Longleftrightarrow\quad
 D_{\rho,j}(H)\text{ exists and }i\notin D_{\rho,j}(H).
 \end{split}
\end{equation}
Indeed, an $E_{\rho,j}$-subline contains $H\cup\{i\}$
if and only if the unique subline containing $H$ also contains
$i$. Thus the local relative degree has the exact factorization
\begin{equation}\label{eq:degree-level-product}
 \begin{aligned}
 [T(H\cup\{i\}):T(H)]
  &=\frac{t_{H\cup\{i\}}}{t_H}
   =\prod_{\rho\mid t}\rho^{|\mathcal J_\rho(H,i)|},\\
 \mathcal J_\rho(H,i)
 &=\{1\le j\le a_\rho:
       D_{\rho,j}(H)\text{ exists},\ i\notin D_{\rho,j}(H)\}\\
  &=\{j\in\mathbb Z:v_\rho(t_H)<j\le v_\rho(t_{H\cup\{i\}})\}.
 \end{aligned}
\end{equation}
The last set is empty when the two valuations are equal.

We next compare $T(A)$ with each field $T(A\setminus\{i\})$
obtained by deleting one point.
Define
\begin{equation}\label{eq:S-N}
 \begin{gathered}
 h_i=t_{A\setminus\{i\}},\qquad S_i=t/h_i,\qquad
 N=\frac{t}{\prod_{i\in A}S_i},\\
 Z=\{i\in A:S_i<s\},\qquad z=|Z|.
 \end{gathered}
\end{equation}
Here $S_i=[T(A):T(A\setminus\{i\})]$ records the change in
field degree when $i$ is deleted, and $Z$ collects the points
for which this ratio is below $s$. We first show that the $S_i$
are pairwise coprime and that $N$ is an integer. If $Z$ is empty,
these facts already give $t\ge Q_s(n)$. Otherwise, we bound $N$
using the helper-set constraints. The prime-divisor estimates in
Appendix~\ref{app:finite} give the finite-parameter result;
Section~\ref{sec:quadratic} gives the quadratic extension.

By \eqref{eq:delete} and the compositum formula in
\eqref{eq:fields}, for distinct $i,i'\in A$,
\[
 t=\lcm(h_i,h_{i'}),\qquad
 a_\rho=\max\{v_\rho(h_i),v_\rho(h_{i'})\}
 \quad(\rho\mid t).
\]
Since $v_\rho(S_i)=a_\rho-v_\rho(h_i)$, this implies
\[
 \min\{v_\rho(S_i),v_\rho(S_{i'})\}=0,
 \qquad \gcd(S_i,S_{i'})=1.
\]
Thus each prime divisor of $t$ divides at most one $S_i$,
and its exponent in that $S_i$ is at most its exponent in $t$.
It follows that $\prod_{i\in A}S_i\mid t$ and $N$ is a
positive integer. The $n-z$ integers $S_i$ with $i\notin Z$
are pairwise coprime and at least $s$, so their product is
at least $Q_s(n-z)$. Hence
\begin{equation}\label{eq:factor-bound}
 t=\left(\prod_{i\in A}S_i\right)N
 \ge Q_s(n-z)N.
\end{equation}
Define the following sets of pairs $(\rho,j)$:
\begin{equation}\label{eq:level-sets}
 \begin{aligned}
 \mathcal L&=\{(\rho,j):\rho\mid t,\ 1\le j\le a_\rho\},\\
 \mathcal S_i&=\{(\rho,j)\in\mathcal L:v_\rho(h_i)<j\},\\
 \mathcal N&=\mathcal L\setminus\bigcup_{i\in A}\mathcal S_i.
 \end{aligned}
\end{equation}
If $(\rho,j)\in\mathcal S_i\cap\mathcal S_{i'}$ with
$i\ne i'$, then
$\max\{v_\rho(h_i),v_\rho(h_{i'})\}<j\le a_\rho$,
contrary to $t=\lcm(h_i,h_{i'})$. Thus the sets
$\mathcal S_i$ are pairwise disjoint, and
\[
 S_i=\prod_{(\rho,j)\in\mathcal S_i}\rho,\qquad
 N=\prod_{(\rho,j)\in\mathcal N}\rho.
\]
By \eqref{eq:threshold}, $(\rho,j)$ belongs to $\mathcal S_i$
exactly when an $E_{\rho,j}$-subline contains $A\setminus\{i\}$
and excludes $i$. Such a subline cannot contain all of $A$,
because $T(A)=L\not\subseteq E_{\rho,j}$.

We will also use the implication
\begin{equation}\label{eq:subline-N}
 \left.
 \begin{array}{c}
 D\text{ is an }E_{\rho,j}\text{-subline},\\
 |A\cap D|\ge d,\quad |A\setminus D|\ge2
 \end{array}
 \right\}
 \quad\Longrightarrow\quad (\rho,j)\in\mathcal N.
\end{equation}
Indeed, if $(\rho,j)\in\mathcal S_i$ for some $i\in A$,
there is an $E_{\rho,j}$-subline $D'$ with
$A\setminus D'=\{i\}$. Then
\[
 |A\cap D\cap D'|\ge d-1\ge3,
\]
so Lemma~\ref{lem:subline-intersection} would give $D=D'$,
contrary to $|A\setminus D|\ge2$ and $|A\setminus D'|=1$.

\subsection{The Finite-Parameter Bound from the First Version}

The following theorem was established in the first version.
Its explicit range includes small parameters not covered by the
asymptotic assertion below. The proof is in Appendix~\ref{app:finite}.

\begin{theorem}\label{thm:main}
For every prime power $q$ and all integers
$3\le k<d<n\le2d-2$,
\begin{equation}\label{eq:main}
 L_q(n,k,d)=sQ_s(n),\qquad s=d-k+1.
\end{equation}
\end{theorem}

\section{Equality below the Quadratic Scale}\label{sec:quadratic}

This section proves the lower side of the quadratic transition.
When only a bounded number of points have deletion ratio below
$s$, a prime-divisor count suffices. Otherwise, we normalize three
such points and extract an arithmetic degree tuple. A finite
covering inequality and its asymptotic consequence give lower
bounds on the remaining prime-power product. Comparing these
bounds with the degree of a minimum-extension code completes
the proof.

\begin{theorem}\label{thm:quadratic}
For every $0<\varepsilon<1$, there is $d_0(\varepsilon)$ such that
for all admissible parameters with $d\ge d_0(\varepsilon)$,
\begin{align}
 n\le(1-\varepsilon)d^2
 &\quad\Longrightarrow\quad
 L_q(n,k,d)=s\Delta_s(n,d)=sQ_s(n),
 \label{eq:quadratic-lower}\\
 n\ge(1+\varepsilon)d^2
 &\quad\Longrightarrow\quad
 L_q(n,k,d)<sQ_s(n).
 \label{eq:quadratic-upper}
\end{align}
The threshold $d_0(\varepsilon)$ is independent of $q,k,s,n$.
The second assertion holds at every integer length in its range.
\end{theorem}

We prove \eqref{eq:quadratic-lower} here; the upper side is proved
in Section~\ref{sec:small}.
Retain $t,S_i,N,Z,z$ from Section~\ref{sec:lower}.
Thus $Z=\{i\in A:S_i<s\}$ and $z=|Z|$; we call these points
deficient.
All point sets below are normalized sets used to study their
fields; the RS codes continue to have finite evaluation points.

\subsection{Residual Factors and Three Reference Points}

We need the following bound without any restriction on $n$
beyond the standing parameters. It is the general residual bound
of \cite[Lemma~8, equation~(17)]{FQXF}, expressed in our notation.
Our quantities $t,h_i,S_i,N,z$ correspond, respectively,
to $M,m_i,g_i,M_0,b$ in that reference.
\begin{lemma}[\cite{FQXF}]\label{lem:general-residual}
If $z>0$, then
\begin{equation}\label{eq:general-residual}
 \omega(N)\ge d-2,\qquad N\ge\Pi_{d-2}\ge2^{d-2}.
\end{equation}
\end{lemma}

Three distinct reference points determine the normalization to
$0,\infty,1$. Choosing them from $Z$ keeps every factor $S_i\ge s$
among the remaining points, preserving the factor $Q_s(n-z)$
in the degree estimate below. The extraction produces an arithmetic
tuple with $D=d-3$ and the original threshold $s$; it does not
introduce a code of dimension $k-3$. In particular, it also applies
when $k=3,4,5$.

\begin{lemma}\label{quad:lem:extraction}
Suppose $z\ge3$, and choose a three-element set $C\subseteq Z$.
Normalize its points to $0,\infty,1$. For $i\notin C$, let $x_i$
be the normalized coordinate and put $a_i=[K(x_i):K]$.
Then $(a_i)_{i\notin C}$ is a $(D,s)$-feasible tuple with
\begin{equation}\label{quad:eq:extractedparameters}
                 D=d-3,\qquad m=n-3,\qquad
                 \lcm_{i\notin C}a_i=t.
\end{equation}
For each $i\notin C$, its singleton factor, defined as the ratio
of the full least common multiple to the one omitting $i$, is
\[
 \sigma_i=\frac{\lcm_{j\notin C}a_j}
 {\lcm_{j\notin C\cup\{i\}}a_j}=S_i.
\]
Consequently, the number of indices with $\sigma_i<s$ and the
residual integer are
\begin{equation}\label{quad:eq:extractedresidual}
 w=z-3,\qquad
 M=\frac{t}{\prod_{i\notin C}S_i}
          =N\prod_{c\in C}S_c,
 \qquad t\ge Q_s(n-z)M.
\end{equation}
\end{lemma}
\begin{proof}
By \eqref{eq:anchors} and \eqref{eq:fields},
\[
 T(A)=K(x_i:i\notin C),\qquad
 t=\lcm_{i\notin C}[K(x_i):K]
   =\lcm_{i\notin C}a_i.
\]
For $i\notin C$ and
$J\subseteq A\setminus(C\cup\{i\})$ with $|J|=D$, put
$H=C\cup J$. Then
\[
 \begin{aligned}
 |H|&=3+D=d,\qquad i\notin H,\\
 T(H)&=K(x_j:j\in J),\qquad
 T(H\cup\{i\})=K(x_i,x_j:j\in J).
 \end{aligned}
\]
Hence \eqref{eq:local} gives
\[
 \frac{\lcm(a_i,a_j:j\in J)}{\lcm(a_j:j\in J)}
 =[T(H\cup\{i\}):T(H)]\ge s.
\]
For $i\notin C$, the same reference points give
\[
 T(A\setminus\{i\})=K(x_j:j\notin C\cup\{i\}),\qquad
 \sigma_i=\frac{t}{[T(A\setminus\{i\}):K]}=S_i.
\]
Since $C\subseteq Z$,
\[
 \begin{aligned}
 \{i\notin C:\sigma_i<s\}&=Z\setminus C,\qquad w=z-3,\\
 M&=\frac{t}{\prod_{i\notin C}S_i}
   =N\prod_{c\in C}S_c.
 \end{aligned}
\]
The factors $S_i$ are pairwise coprime and satisfy $S_i\ge s$
for $i\notin Z$. Therefore
\[
 t=M\prod_{i\notin C}S_i
   \ge M\prod_{i\notin Z}S_i
   \ge M Q_s(n-z).
\]
\end{proof}

\subsection{A Finite Cost Inequality for Degree Tuples}

In this subsection $m>D\ge1$ and $s\ge2$ are arbitrary arithmetic
parameters. Let $(a_1,\ldots,a_m)$ be $(D,s)$-feasible, and write
\begin{equation}\label{quad:eq:arithnotation}
 b_0=\lcm_i a_i,\quad
 \sigma_i=\frac{b_0}{\lcm_{j\ne i}a_j},\quad
 M=\frac{b_0}{\prod_i\sigma_i},\quad
 Z_0=\{i:\sigma_i<s\},\quad w=|Z_0|.
\end{equation}
For each prime $p\mid b_0$ and $1\le r\le v_p(b_0)$, define
the threshold support
\begin{equation}\label{quad:eq:tuplesupports}
                     U_{p,r}=\{i:v_p(a_i)\ge r\}.
\end{equation}
Call $(p,r)$ a \emph{residual level} if $|U_{p,r}|\ge2$,
and call $U_{p,r}$ its residual support. Put
\[
 \mathcal R_p=\{1\le r\le v_p(b_0):|U_{p,r}|\ge2\}.
\]
The supports at a fixed prime are nested. The singleton levels
give exactly the factors $\sigma_i$, so
\begin{equation}\label{quad:eq:level-factorization}
 \sigma_i=\prod_{p\mid b_0}p^{|\{r:U_{p,r}=\{i\}\}|},
 \qquad
 b_0=\left(\prod_i\sigma_i\right)
          \prod_{p\mid b_0}p^{|\mathcal R_p|}.
\end{equation}
Indeed, the singleton levels at $i$ are precisely the integers
strictly above $\max_{j\ne i}v_p(a_j)$ and at most $v_p(b_0)$.
At a fixed prime they can occur for at most one index.
Thus the $\sigma_i$ are pairwise coprime,
$\prod_i\sigma_i\mid b_0$, and $M$ is a positive integer.
For every prime $p$, write $b_p=v_p(M)$; for $p\mid b_0$,
the factorization gives $b_p=|\mathcal R_p|$.

For every $i$ and every nonempty
$J\subseteq\{1,\ldots,m\}\setminus\{i\}$, we have
\begin{equation}\label{quad:eq:level-ratio}
 \frac{\lcm(a_i,a_j:j\in J)}{\lcm(a_j:j\in J)}
 =\sigma_i\prod_{p\mid M}
 p^{|\{r\in\mathcal R_p:
       i\in U_{p,r},\ U_{p,r}\cap J=\varnothing\}|}.
\end{equation}
The exponent of $p$ on the left is
\[
 \max\left\{v_p(a_i)-\max_{j\in J}v_p(a_j),0\right\},
\]
which counts exactly the levels containing $i$ and disjoint from
$J$. Every singleton level at $i$ is disjoint from $J$;
the remaining levels are those counted in the product on the right.

\subsubsection{Weighted Coverings and Projection}

For a finite indexed family of nonempty sets $(C_j)_{j\in J}$, its
transversal number is
\[
 \tau((C_j))=\min\{|X|:X\cap C_j\ne\varnothing\text{ for every }j\}.
\]
Repeated sets retain their indices. The empty family has transversal
number zero.

\begin{lemma}[{\cite[Lemma~2.2]{HZ}}]\label{quad:lem:HZ}
Let $\mathcal H$ be a finite family covering a nonempty finite set
$U$. Give its members probability weights $\mu_H\ge0$,
$\sum_{H\in\mathcal H}\mu_H=1$. If
$p_x=\sum_{H\ni x}\mu_H>0$ for every $x\in U$, let $\kappa(\mathcal H)$ be the minimum number of its members
whose union is $U$. Then
\begin{equation}\label{quad:eq:HZ}
                   \kappa(\mathcal H)^2\le\sum_{x\in U}p_x^{-1}.
\end{equation}

\end{lemma}

\begin{lemma}\label{quad:lem:capped}
Let $(C_j)_{j\in J}$ be nonempty subsets of a finite universe,
with $\tau((C_j))\ge R\ge1$. Let $W$ be a subset of that universe,
$V=|W|>0$, and $e_j=|C_j\cap W|$. Set
\begin{equation}\label{quad:eq:weight}
 \varphi_{R,V}(e)=
 \begin{cases}
 2R,&e=0,\\
 \min\{2R,V/e\},&e\ge1.
 \end{cases}
\end{equation}
Then
\begin{equation}\label{quad:eq:capped}
                         \sum_{j\in J}\varphi_{R,V}(e_j)\ge R^2.
\end{equation}
\end{lemma}
\begin{proof}
Partition the labels into
\[
 J_0=\{j\in J:e_j\le V/(2R)\},\qquad
 J_1=J\setminus J_0,\qquad a=|J_0|.
\]
For $j\in J_0$, $\varphi_{R,V}(e_j)=2R$. If $a\ge R$,
then the desired sum is at least $2Ra\ge R^2$.
Assume $a<R$. Choose $x_j\in C_j$ for each $j\in J_0$ and put
$T_0=\{x_j:j\in J_0\}$, so $|T_0|\le a$.
Every $T\subseteq W$ meeting all $C_j\cap W$ with $j\in J_1$
satisfies
\[
 (T\cup T_0)\cap C_j\ne\varnothing\quad(j\in J),\qquad
 R\le\tau((C_j))\le |T\cup T_0|\le |T|+a.
\]
Hence $\tau((C_j\cap W)_{j\in J_1})\ge R-a>0$;
in particular, $J_1\ne\varnothing$.

For $x\in W$, define
\[
 H_x=\{j\in J_1:x\in C_j\},\qquad
 \mathcal H=\{H_x:x\in W\}.
\]
Because $e_j>0$ for $j\in J_1$,
\[
 \bigcup_{x\in W}H_x=J_1,\qquad
 \kappa(\mathcal H)
   =\tau((C_j\cap W)_{j\in J_1})\ge R-a.
\]
Give each distinct $H\in\mathcal H$ weight
\[
 \mu_H=\frac{|\{x\in W:H_x=H\}|}{V}.
\]
Then
\[
 \sum_{H\in\mathcal H}\mu_H=1,\qquad
 p_j=\sum_{H\ni j}\mu_H=\frac{e_j}{V}>0.
\]
Applying \eqref{quad:eq:HZ} gives
\[
 (R-a)^2\le\kappa(\mathcal H)^2
              \le V\sum_{j\in J_1}e_j^{-1}.
\]
Since $\varphi_{R,V}(e_j)=V/e_j$ for $j\in J_1$,
\[
 \sum_{j\in J}\varphi_{R,V}(e_j)
 =2Ra+V\sum_{j\in J_1}e_j^{-1}
 \ge2Ra+(R-a)^2=R^2+a^2\ge R^2.
\]
\end{proof}

\subsubsection{Prime Chains and Their Cost}

\begin{theorem}\label{quad:thm:capacity}
For every integer
\begin{equation}\label{quad:eq:hcondition}
                  2\le h\le\min(D+1,w-1),
\end{equation}
put
\begin{equation}\label{quad:eq:capacityparameters}
 R=D-h+2,\qquad V=w-h,\qquad
 B_h=\frac{2R}{h},\qquad
 \Lambda_h=\frac{w-h+2R}{w}.
\end{equation}
Then
\begin{equation}\label{quad:eq:finitecapacity}
                   \sum_{p\mid M}\min\{B_h,\Lambda_h b_p\}\ge R^2.
\end{equation}
In particular, for every real $x>1$,
\begin{equation}\label{quad:eq:finitecost}
 \log M\ge
     \frac{\pos{R^2-B_h\pi(x)}}{\Lambda_h}\log x.
\end{equation}
There is no restriction on the prime bases or their exponents.
\end{theorem}
In the proof, $B_h$ bounds the average contribution of one prime,
while $\Lambda_h$ bounds the average contribution of one residual
level. A prime with $b_p$ residual levels therefore has the two
bounds $B_h$ and $b_p\Lambda_h$, before the final conversion to
$\log M$.

\begin{proof}
Fix $C\subseteq Z_0$ with $|C|=h$ and $i\in C$.
We require the helper set to contain $C\setminus\{i\}$.
Define the set of eligible primes by
\[
 \mathcal E(C,i)=
 \{p\mid M:\exists r\in\mathcal R_p,\ U_{p,r}\cap C=\{i\}\}.
\]
For $p\in\mathcal E(C,i)$, put
\[
 r_p(i)=\max\{r\in\mathcal R_p:i\in U_{p,r}\},\qquad
 U_p(i)=U_{p,r_p(i)},\qquad C_p(i)=U_p(i)\setminus\{i\}.
\]
The nesting of the supports gives
\[
 \begin{aligned}
 &\varnothing\ne C_p(i)\subseteq\{1,\ldots,m\}\setminus C,\\
 &C_p(i)\subseteq U_{p,r}\setminus\{i\}
       \quad(r\in\mathcal R_p,\ i\in U_{p,r}).
 \end{aligned}
\]

Suppose there were a set
$T\subseteq\{1,\ldots,m\}\setminus C$ such that
\[
 |T|\le D-h+1,\qquad
 T\cap C_p(i)\ne\varnothing\quad(p\in\mathcal E(C,i)).
\]
Since $m>D$, we could choose $J$ satisfying
\[
 T\cup(C\setminus\{i\})\subseteq J
 \subseteq\{1,\ldots,m\}\setminus\{i\},\qquad |J|=D.
\]
For an eligible prime, $T$ meets every residual support containing
$i$. For a noneligible prime, each such support meets
$C\setminus\{i\}$. Therefore
\[
 i\in U_{p,r}\quad\Longrightarrow\quad
 U_{p,r}\cap J\ne\varnothing
       \qquad(p\mid M,\ r\in\mathcal R_p).
\]
By \eqref{quad:eq:level-ratio},
\[
 \frac{\lcm(a_i,a_j:j\in J)}{\lcm(a_j:j\in J)}
       =\sigma_i<s,
\]
contrary to \eqref{eq:tuple-feasible}. Hence
\begin{equation}\label{quad:eq:blockertau}
 \tau((C_p(i))_{p\in\mathcal E(C,i)})\ge D-h+2=R.
\end{equation}
Apply Lemma~\ref{quad:lem:capped} with $W=Z_0\setminus C$:
\[
 \sum_{p\in\mathcal E(C,i)}
 \varphi_{R,V}\bigl(|C_p(i)\cap(Z_0\setminus C)|\bigr)
 \ge R^2.
\]

For a fixed $C$, a prime is eligible for at most one target:
\[
 \bigl|\{i\in C:p\in\mathcal E(C,i)\}\bigr|\le1.
\]
Indeed, two supports in the same chain cannot intersect $C$
in different singletons. Define
\[
 X_p=\frac{1}{h\binom wh}
 \sum_{\substack{C\subseteq Z_0\\|C|=h}}
 \sum_{\substack{i\in C\\p\in\mathcal E(C,i)}}
 \varphi_{R,V}\bigl(|C_p(i)\cap(Z_0\setminus C)|\bigr).
\]
Averaging the preceding inequalities and using
$\varphi_{R,V}\le2R$ yields
\begin{equation}\label{quad:eq:primewholecap}
 \sum_{p\mid M}X_p\ge R^2,\qquad
 0\le X_p\le\frac{2R}{h}=B_h.
\end{equation}

For each residual level $r\in\mathcal R_p$, define
\[
 Y_{p,r}=\frac{1}{h\binom wh}
 \sum_{\substack{C\subseteq Z_0,\ |C|=h\\
                 |C\cap U_{p,r}|=1}}
 \varphi_{R,V}\bigl(|U_{p,r}\cap(Z_0\setminus C)|\bigr).
\]
Put $v=|U_{p,r}\cap Z_0|$. If $v=0$, then $Y_{p,r}=0$.
If $v\ge1$, there are $v\binom{w-v}{h-1}$ sets $C$ in this sum,
and each summand is $\varphi_{R,V}(v-1)$.
Here $\binom ab=0$ for integers $0\le a<b$. Thus
\begin{equation}\label{quad:eq:levelaverage}
 \begin{aligned}
 Y_{p,r}
 &=\frac{v\binom{w-v}{h-1}}{h\binom wh}
                         \varphi_{R,V}(v-1)\\
 &=\frac vw\,
   \frac{\binom{w-v}{h-1}}{\binom{w-1}{h-1}}
                         \varphi_{R,V}(v-1)\\
 &\le\frac vw\varphi_{R,V}(v-1).
 \end{aligned}
\end{equation}
For every integer $v\ge1$,
\begin{equation}\label{quad:eq:vweight}
 v\varphi_{R,V}(v-1)
 =
 \begin{cases}
  2Rv,&2R(v-1)\le V,\\
  V+V/(v-1),&2R(v-1)>V
 \end{cases}
 \le V+2R.
\end{equation}
Consequently,
\[
                  Y_{p,r}\le\frac{V+2R}{w}=\Lambda_h.
\]

For every $p\in\mathcal E(C,i)$, its selected level satisfies
\[
 \begin{aligned}
 U_{p,r_p(i)}\cap C&=\{i\},\\
 C_p(i)\cap(Z_0\setminus C)
   &=U_{p,r_p(i)}\cap(Z_0\setminus C).
 \end{aligned}
\]
Its summand in $X_p$ therefore occurs in $Y_{p,r_p(i)}$.
Since each prime has at most one eligible target for a fixed $C$,
\[
 X_p\le\sum_{r\in\mathcal R_p}Y_{p,r}
       \le|\mathcal R_p|\Lambda_h=b_p\Lambda_h.
\]
Together with \eqref{quad:eq:primewholecap}, this proves
\eqref{quad:eq:finitecapacity}.

Finally, for $x>1$, split the prime sum at $x$:
\[
 \begin{aligned}
 R^2
 &\le\sum_{\substack{p\mid M\\p\le x}}B_h
       +\sum_{\substack{p\mid M\\p>x}}\Lambda_h b_p\\
 &\le B_h\pi(x)
       +\Lambda_h\sum_{\substack{p\mid M\\p>x}}b_p,\\
 \log M
 &=\sum_{p\mid M}b_p\log p
 \ge\log x\sum_{\substack{p\mid M\\p>x}}b_p.
 \end{aligned}
\]
Combining these inequalities with $\log M\ge0$ gives
\eqref{quad:eq:finitecost}.
\end{proof}

The next consequence applies when the number $w$ of indices with
$\sigma_i<s$ is at least a fixed positive multiple of $D^2$.
It will handle the third case in the quadratic lower-bound proof.

\begin{corollary}\label{quad:cor:quadraticcost}
Fix $\alpha>0$. Uniformly over all $(D,s)$-feasible tuples with
$w\ge\alpha D^2$, all $s\ge2$, and all tuple lengths,
\begin{equation}\label{quad:eq:quadraticcost}
                       \log M\ge(1-o_\alpha(1))D^2\log(D^2)
                       \qquad(D\longrightarrow\infty).
\end{equation}
\end{corollary}
\begin{proof}
In Theorem~\ref{quad:thm:capacity}, choose
\[
                h=\lfloor D/\log D\rfloor,
                \qquad x=\frac{D^2}{(\log D)^3}.
\]
These choices satisfy \eqref{quad:eq:hcondition} for all sufficiently
large $D$, uniformly for $w\ge\alpha D^2$. Also
\begin{align*}
 R^2&=(1-O(1/\log D))D^2,\\
 B_h&=O(\log D),\\
 \Lambda_h&=1+\frac{2D-3h+4}{w}=1+O_\alpha(D^{-1}),\\
 B_h\pi(x)&\le B_hx=O(D^2/(\log D)^2)=o(D^2),\\
 \log x&=2\log D-3\log\log D=(1-o(1))\log(D^2).
\end{align*}
Substitution in \eqref{quad:eq:finitecost} proves the assertion.
All error bounds depend only on $\alpha$ and $D$.
\end{proof}

\subsection{Completion of the Quadratic Lower Bound}

We use the following standard consequences of the prime number
theorem; see, for example,~\cite{Dusart}:
\begin{equation}\label{eq:quadratic-prime-facts}
 p_j\sim j\log j,\qquad
 \log p_{d^2+d}=(2+o(1))\log d,\qquad
 r(d):=\min\{p>d:p\text{ prime}\}=(1+o(1))d.
\end{equation}
For the last assertion, each fixed $\delta>0$ satisfies
$\pi((1+\delta)d)-\pi(d)=(\delta+o(1))d/\log d>0$ for
sufficiently large $d$.

\begin{proof}[Proof of the lower bound in Theorem~\ref{thm:quadratic}]
Fix $0<\varepsilon<1$ and suppose
$n\le(1-\varepsilon)d^2$. Corollary~\ref{cor:tuple-realization}
supplies the upper bound $L_q(n,k,d)\le sQ_s(n)$, so the minimum
exists. Choose a code attaining it. By \eqref{eq:intrinsic},
\begin{equation}\label{quad:eq:optimalbudget}
                             t\le L_q(n,k,d)/s\le Q_s(n).
\end{equation}
We prove that $z=0$.

Set $\mathfrak p_d=p_{d^2+d}$. Since $n\le d^2$ and $s\le d-2$,
\begin{equation}\label{quad:eq:uniformprimebudget}
 p_{n+\pi(s-1)}\le \mathfrak p_d,
 \qquad \log \mathfrak p_d=(2+o(1))\log d.
\end{equation}
From \eqref{eq:factor-bound}, \eqref{eq:Qgrowth}, and
\eqref{quad:eq:optimalbudget}, if $z>0$ then
\[
 N\le\frac{t}{Q_s(n-z)}
   \le\frac{Q_s(n)}{Q_s(n-z)}
   \le p_{n+\pi(s-1)}^z\le\mathfrak p_d^z.
\]
Therefore
\begin{equation}\label{quad:eq:Nupper}
                             \log N\le z\log \mathfrak p_d.
\end{equation}
If $z\ge3$, choose three reference points from $Z$ and apply
Lemma~\ref{quad:lem:extraction}. Since
\[
 M\le\frac{t}{Q_s(n-z)}\le\mathfrak p_d^z,
\]
we obtain
\begin{equation}\label{quad:eq:Mupper}
                             \log M\le z\log \mathfrak p_d,
                    \qquad D=d-3,\quad w=z-3.
\end{equation}

For $z\le5$ we use the general prime-divisor bound.
For $z\ge6$, the extracted tuple has $w=z-3\ge3$, allowing
$h=2$ in Theorem~\ref{quad:thm:capacity}.
When $w$ is proportional to $d^2$, we instead use
Corollary~\ref{quad:cor:quadraticcost}.
The intermediate cutoff $d^2/16$ is chosen to leave a fixed
gap between the estimates in the second case.

\smallskip\noindent
\emph{Case 1: $1\le z\le5$.}
Lemma~\ref{lem:general-residual} and \eqref{quad:eq:Nupper} imply
\[
                    (d-2)\log2\le\log N
                                  \le5\log \mathfrak p_d=O(\log d),
\]
which is impossible for sufficiently large $d$.

\smallskip\noindent
\emph{Case 2: $6\le z\le d^2/16$.}
Use the extracted tuple and take $h=2$, $x=\sqrt D$ in
Theorem~\ref{quad:thm:capacity}. Now $w\ge3$, so the choice of $h$ is admissible,
and
\[
        R=D,\qquad B_h=D,\qquad
        \Lambda_h=\frac{w+2D-2}{w},\qquad \pi(\sqrt D)\le\sqrt D.
\]
Consequently
\begin{equation}\label{quad:eq:mediumcost}
 \log M\ge
        \frac{(D^2-D^{3/2})w}{w+2D-2}\,\frac{\log D}{2}.
\end{equation}
The bounds $6\le z\le d^2/16$ give $w=z-3\ge z/2$ and
$w+2D-2\le d^2/16+2D$. Therefore
\begin{equation}\label{quad:eq:mediumratio}
 \frac{\log M}{z\log d}
 \ge\frac{D^2-D^{3/2}}{4(d^2/16+2D)}\frac{\log D}{\log d}
                          =4-o(1).
\end{equation}
The error term is independent of $z$. But
\eqref{quad:eq:Mupper}--\eqref{quad:eq:uniformprimebudget} give
$\log M/(z\log d)\le2+o(1)$, a contradiction.

\smallskip\noindent
\emph{Case 3: $z>d^2/16$.}
Since $z>d^2/16$ and $D=d-3$, for sufficiently large $d$,
$w=z-3>d^2/16-3\ge D^2/32$.
Corollary~\ref{quad:cor:quadraticcost} with
$\alpha=1/32$ yields
\begin{equation}\label{quad:eq:largecost}
                         \log M\ge(2-o(1))d^2\log d.
\end{equation}
In contrast, \eqref{quad:eq:Mupper} and $z\le n\le(1-\varepsilon)d^2$
give
\begin{equation}\label{quad:eq:largeupper}
 \log M\le(1-\varepsilon)d^2(2+o(1))\log d.
\end{equation}
For sufficiently large $d$, the two error terms have absolute
value at most $\varepsilon/2$. Hence
\[
 2-\frac{\varepsilon}{2}
 \le\frac{\log M}{d^2\log d}
 \le(1-\varepsilon)\left(2+\frac{\varepsilon}{2}\right)
 <2-\frac{\varepsilon}{2},
\]
a contradiction.

All positive values of $z$ have been excluded. Thus $z=0$, and
Equation~\eqref{eq:factor-bound} gives $t\ge Q_s(n)$. Finally,
\[
                 sQ_s(n)\le st\le L_q(n,k,d)\le sQ_s(n).
\]
Together with \eqref{eq:tuple-upper}, this also gives
$s\Delta_s(n,d)=sQ_s(n)$.
Every estimate used above is uniform in $q$, in $2\le s\le d-2$,
and in the permitted lengths $n$. Only the final comparison in
Case 3 depends on $\varepsilon$.
\end{proof}

\section{Strict Improvement and the Transition Length}
\label{sec:small}

This section combines the preceding lower bound with the
affine-plane bound of Fang, Qiu, Xia, and Fu to locate the first
length admitting a strict improvement. We also give polynomial-graph
bounds for selected parameter families.

\subsection{An Affine-Plane Upper Bound}

We use the following result in the notation of
Corollary~\ref{cor:tuple-realization}.
\begin{proposition}[{\cite[Proposition~5]{FQXF}}]\label{prop:affine-upper}
Let $r>d$ be a prime power. For every $n\ge r^2+1$, put
$v=\max\{r^2,n-r\}$. Then
\begin{equation}\label{eq:affine-upper}
 L_q(n,k,d)\le s\Delta_s(n,d)
 \le sQ_s(v)<sQ_s(n).
\end{equation}
The auxiliary prime power $r$ is independent of the repair
base-field size $q$.
\end{proposition}

\begin{proof}[Proof of the upper bound in Theorem~\ref{thm:quadratic}]
Let $r(d)$ be the least prime greater than $d$.
The prime number theorem gives $r(d)=(1+o(1))d$; see
\eqref{eq:quadratic-prime-facts}.
For each fixed $\varepsilon>0$ and sufficiently large $d$,
\[
 r(d)^2+1\le(1+\varepsilon)d^2.
\]
For every integer $n\ge(1+\varepsilon)d^2$, it follows that
\[
 n\ge(1+\varepsilon)d^2\ge r(d)^2+1,\qquad
 v:=\max\{r(d)^2,n-r(d)\}\le n-1.
\]
Proposition~\ref{prop:affine-upper} therefore gives
\[
 L_q(n,k,d)\le sQ_s(v)<sQ_s(n).
\]
\end{proof}

\begin{corollary}\label{cor:transition-length}
For each $q,k,d$, define
\[
 n_*(q,k,d)=\min\{n>d:L_q(n,k,d)<sQ_s(n)\}.
\]
This minimum exists, and
\begin{equation}\label{eq:transition-length}
 n_*(q,k,d)=(1+o(1))d^2\qquad(d\longrightarrow\infty),
\end{equation}
uniformly over all prime powers $q$ and all $3\le k<d$.
\end{corollary}
\begin{proof}
Existence follows from Proposition~\ref{prop:affine-upper}.
For each fixed $0<\varepsilon<1$, Theorem~\ref{thm:quadratic}
gives, for all sufficiently large $d$,
\[
 \lfloor(1-\varepsilon)d^2\rfloor+1
 \le n_*(q,k,d)\le\lceil(1+\varepsilon)d^2\rceil.
\]
Dividing by $d^2$ proves the assertion.
\end{proof}

Thus the approximate transition occurs at $d^2$.
The result does not determine the exact boundary inside
$n=d^2+o(d^2)$ or the exact minimum above it.

\subsection{Polynomial-Graph Bounds from the First Version}

The following bounds were established in the first version.
They give explicit degrees for selected parameter families and
can be smaller than the general affine-plane upper bound.
Their proof uses polynomial graphs to construct cover-free families.

\begin{lemma}\label{lem:coverfree}
Suppose that there exist sets $B_1,\ldots,B_n\subseteq\{1,\ldots,m\}$ such that
\begin{equation}\label{eq:coverfree}
 B_i\not\subseteq\bigcup_{j\in J}B_j
\end{equation}
for every $i\in\{1,\ldots,n\}$ and every
$J\subseteq\{1,\ldots,n\}\setminus\{i\}$ with $|J|=d$.
Then $L_q(n,k,d)\le sQ_s(m)$ for every prime power $q$.
\end{lemma}

\begin{proof}
Choose pairwise coprime integers $a_1,\ldots,a_m\ge s$
with product $b=Q_s(m)$, and set
$f_i=\prod_{u\in B_i}a_u$.
Condition~\eqref{eq:coverfree} forces the sets $B_i$ to be
distinct and nonempty. Unique prime factorization and the
coprimality of the $a_u$ give distinct divisors $f_i$ of $b$,
and, for every such $i$ and $J$,
\begin{equation}\label{eq:coverfree-degrees}
 \frac{f_i}{\gcd(f_i,\lcm_{j\in J}f_j)}
 =\prod_{u\in B_i\setminus\bigcup_{j\in J}B_j}a_u\ge s.
\end{equation}
The final inequality holds because
$B_i\setminus\bigcup_{j\in J}B_j\ne\varnothing$ by
\eqref{eq:coverfree} and every $a_u\ge s$.
Thus the distinct divisors $f_1,\ldots,f_n$ of $b$ satisfy
\eqref{eq:degree-condition}. Proposition~\ref{prop:exist} gives
an RS code over $\F_{q^{sb}}$ attaining the cut-set bound for
every repair instance. Since $b=Q_s(m)$, this proves
$L_q(n,k,d)\le sQ_s(m)$.
\end{proof}

\begin{theorem}\label{thm:graphs}
Let $R$ be a prime power and $v\ge1$ an integer. If
\begin{equation}\label{eq:graph-parameters}
 dv+1\le R,\qquad R(dv+1)<n\le R^{v+1},
\end{equation}
then, for every prime power $q$,
\begin{equation}\label{eq:graphs}
 L_q(n,k,d)\le sQ_s\bigl(R(dv+1)\bigr)<sQ_s(n).
\end{equation}
In particular, for every prime power $R\ge d+2$,
\begin{equation}\label{eq:graphs-special}
 L_q(R^2,k,d)\le sQ_s(R(d+1))<sQ_s(R^2).
\end{equation}
\end{theorem}

\begin{proof}
Put $m_0=dv+1$ and choose distinct $x_1,\ldots,x_{m_0}\in\F_R$.
For each $f\in\F_R[X]$ of degree at most $v$, let
\[
 B_f=\{(a,f(x_a)):1\le a\le m_0\}
 \subseteq\{1,\ldots,m_0\}\times\F_R.
\]
This is the polynomial-graph form of the classical
superimposed-code construction~\cite{KS}.
There are $R^{v+1}$ distinct sets on a ground set of size
$m=m_0R$, and
\[
 |B_f|=m_0,
 \qquad |B_f\cap B_g|\le v\quad(f\ne g).
\]
Thus the union of any $d$ other sets covers at most $dv<m_0$
points of $B_f$. Any $n$ of these sets satisfy
\eqref{eq:coverfree}, so Lemma~\ref{lem:coverfree} gives
the first inequality in \eqref{eq:graphs}.
For $m<n$, deleting $n-m$ factors from a minimizing tuple
in \eqref{eq:Q} shows that
\[
 Q_s(n)\ge s^{n-m}Q_s(m)>Q_s(m).
\]
Finally take $v=1$ and $n=R^2$ to obtain
\eqref{eq:graphs-special}.
\end{proof}

The auxiliary field size $R$ is independent of $q$.
These first-version bounds give explicit degrees within the
strict-improvement region. Corollary~\ref{cor:transition-length}
identifies the asymptotic location at which strict
improvement first becomes possible.

\section{Nonlinear Repair at the Cut-Set Bound}
\label{sec:nonlinear}

This section proves prime-field linearization at cut-set equality
and applies it to the extension-degree bounds.
We use the partition criterion of
\cite[Observation~2.1]{CT} and the fiber-counting step in the
proof of \cite[Theorem~5.6]{CT}. The latter concerns symmetric
repair over a prime symbol field; the equality argument over
arbitrary finite extensions is given below.

Write \(q=p^e\), where \(p\) is prime, and put
\[
 F=\mathbb F_{q^\ell}=\mathbb F_{p^m},\qquad m=e\ell,
 \qquad s=d-k+1.
\]
For a fixed repair instance \((c,H)\), let
\(\mu_x:F\to\mathcal Y_x=\mu_x(F)\), \(x\in H\), be deterministic
helper maps. For \(y\in\mathcal Y_x\), its fiber is
\(\mu_x^{-1}(y)=\{u\in F:\mu_x(u)=y\}\).
Using the bandwidth convention \eqref{eq:repair-cost}, write
\begin{equation}\label{eq:nl-cost}
 a_x=|\mathcal Y_x|,\qquad
 B(\mu)=\sum_{x\in H}\log_q a_x.
\end{equation}
The cut-set bound gives \(B(\mu)\ge d\ell/s\).
Neither the individual costs \(\log_q a_x\) nor their sum is
assumed to be an integer.

\begin{theorem}\label{thm:nonlinear}
Suppose that an exact repair scheme for \((c,H)\) satisfies
\(B(\mu)=d\ell/s\). For every \(x\in H\), there is an
\(\mathbb F_p\)-subspace \(N_x\subseteq F\) such that the fibers of
\(\mu_x\) are precisely the cosets of \(N_x\), and
\begin{equation}\label{eq:nl-quotient}
 \dim_{\mathbb F_p}(F/N_x)=\frac{m}{s}=\frac{e\ell}{s}.
\end{equation}
In particular, \(s\mid m\). After relabeling the helper outputs, the
helper maps are \(\mathbb F_p\)-linear and admit an
\(\mathbb F_p\)-linear decoder.
\end{theorem}

\begin{proof}
Put \(Q=|F|\). For each \(J\subseteq H\) with \(|J|=s\),
\eqref{eq:alphabet-cutset} gives
\[
 \delta_J:=\sum_{x\in J}\log_q a_x-\ell\ge0.
\]
Their average satisfies
\[
 \frac{1}{\binom{d}{s}}\sum_{\substack{J\subseteq H\\|J|=s}}\delta_J
 =\frac{s}{d}B(\mu)-\ell=0.
\]
Thus every \(\delta_J=0\). Comparing two such subsets that differ
in one helper gives
\begin{equation}\label{eq:nl-equal-alphabets}
 a_x=a=Q^{1/s}\quad(x\in H),\qquad
 \prod_{x\in J}a_x=Q\quad(|J|=s).
\end{equation}
Since \(a\) is an integer and \(a^s=Q=p^m\), unique prime
factorization gives \(s\mid m\) and \(a=p^{m/s}\).
Put \(t_0=Q/a=a^{s-1}\).

For subsets \(U,V\subseteq F\), write
\(U+V=\{u+v:u\in U,\ v\in V\}\) for their sumset.
Choose \(S\subseteq H\), \(|S|=k\), and a fiber
\(D_x=\mu_x^{-1}(y_x)\), \(y_x\in\mathcal Y_x\), for each
\(x\in S\). Interpolation identifies the possible
target values with
\begin{equation}\label{eq:nl-target-sumset}
 W=\sum_{x\in S}\lambda_xD_x,\qquad
 \lambda_x=\prod_{z\in S\setminus\{x\}}\frac{c-z}{x-z}\ne0.
\end{equation}
With \(y_S=(y_x)_{x\in S}\) fixed, let
\[
 \begin{aligned}
 \Omega&=\left\{
  (\mu_x(f(x)))_{x\in H\setminus S}:
  f\in F[X]_{<k},\ f(x)\in D_x\ (x\in S)
 \right\}\\
 &\subseteq\prod_{x\in H\setminus S}\mathcal Y_x.
 \end{aligned}
\]
Every polynomial defining an element of $\Omega$ produces
the fixed messages $y_S$ on $S$.
If two such polynomials define the same element of $\Omega$,
all their helper messages agree, so exact repair forces their
values at $c$ to agree. Thus the map
\[
 \operatorname{Dec}_{y_S}:\Omega\longrightarrow W,\qquad
 (\mu_x(f(x)))_{x\in H\setminus S}\longmapsto f(c)
\]
is well-defined. Interpolation shows that this map is surjective.
Since \(d-k=s-1\),
\[
 |W|\le|\Omega|\le\prod_{x\in H\setminus S}a_x=a^{s-1}=t_0.
\]
For each \(i\in S\), fixing one element in each \(D_x\),
\(x\in S\setminus\{i\}\), shows that \(W\) contains a
translate of \(\lambda_iD_i\). Hence
\(|D_i|\le|W|\le t_0\). Since \(S\) and its fibers were
arbitrary, this holds for every fiber of every helper.
Each helper partitions \(F\) into \(a\) nonempty fibers, so
\[
 \sum_{y\in\mathcal Y_x}|\mu_x^{-1}(y)|=Q=at_0,
 \qquad |\mu_x^{-1}(y)|\le t_0.
\]
Every term must therefore equal \(t_0\):
\begin{equation}\label{eq:nl-fiber-size}
 |\mu_x^{-1}(y)|=t_0
 \qquad(x\in H,\ y\in\mathcal Y_x).
\end{equation}
For distinct \(i,j\in S\), choose \(u_x\in D_x\) for
\(x\in S\setminus\{i,j\}\) and put
\[
 \gamma=\sum_{x\in S\setminus\{i,j\}}\lambda_xu_x.
\]
Then \(\gamma+(\lambda_iD_i+\lambda_jD_j)\subseteq W\).
Since translation preserves cardinality, we obtain
\begin{equation}\label{eq:nl-no-expansion}
 t_0=|\lambda_iD_i|=|\lambda_jD_j|
 \le|\lambda_iD_i+\lambda_jD_j|\le|W|\le t_0.
\end{equation}

If nonempty subsets \(U,V\subseteq F\) satisfy
\(|U+V|=|U|=|V|\), then they are cosets of the same
\(\mathbb F_p\)-subspace. Indeed, fix \(v_0\in V\).
All translates \(U+v\),
\(v\in V\), equal \(U+v_0\), so
\[
 G=\Span_{\mathbb F_p}(V-v_0),\qquad U+G=U,
 \qquad V\subseteq v_0+G,
 \qquad |V|\le|G|\le|U|.
\]
The last inequalities are equalities, proving the claim.
Apply this fact to \eqref{eq:nl-no-expansion}.
For any \(i\in H\), choose \(S\) containing \(i\) and
\(j\in S\setminus\{i\}\). Fixing \(D_j\) and varying \(D_i\)
keeps the common subspace
\(G=\Span_{\mathbb F_p}\bigl(\lambda_j(D_j-D_j)\bigr)\) fixed.
Hence every fiber of \(\mu_i\) is a coset of the
\(\mathbb F_p\)-subspace \(N_i=\lambda_i^{-1}G\).
The fiber containing zero is \(N_i\) itself, so
\[
 N_i=\mu_i^{-1}(\mu_i(0)),\qquad
 \mu_i(u)=\mu_i(v)\quad\Longleftrightarrow\quad u-v\in N_i
 \quad(u,v\in F).
\]
Thus \(N_i\) is independent of the auxiliary choices. Moreover,
\[
 |F/N_i|=a=p^{m/s},
\]
which proves \eqref{eq:nl-quotient}.

For each helper, define the bijection
\(r_x:\mathcal Y_x\to F/N_x\) by
\(r_x(\mu_x(u))=u+N_x\). This is well-defined because the fibers
are precisely the cosets of \(N_x\). The joint download map on
\(\mathcal C=\operatorname{RS}_k(A)\) is
\[
 \mathsf M:\mathcal C\longrightarrow\prod_{x\in H}F/N_x,
 \qquad \mathsf M(v)=(v_x+N_x)_{x\in H},
\]
which is \(\mathbb F_p\)-linear. Let \(\pi_c(v)=v_c\).
If \(v\in\ker \mathsf M\), then \(v\) and the zero codeword give the
same messages, so exact repair forces \(v_c=0\). Thus
\(\ker \mathsf M\subseteq\ker\pi_c\), and
\[
 D:\operatorname{im}\mathsf M\longrightarrow F,
 \qquad D(\mathsf M(v))=v_c
\]
is well-defined and \(\mathbb F_p\)-linear. Extending \(D\)
linearly to \(\prod_{x\in H}F/N_x\) gives the claimed decoder.
\end{proof}

\begin{corollary}\label{cor:nonlinear}
For a prime base field, let $L_p^{\mathrm{NL}}(n,k,d)$ denote the
minimum extension degree when arbitrary deterministic helper maps
are allowed, with bandwidth measured by \eqref{eq:repair-cost}.
For all admissible parameters,
\begin{equation}\label{eq:nl-prime-equality}
 L_p^{\mathrm{NL}}(n,k,d)=L_p(n,k,d).
\end{equation}
Consequently, the finite-parameter formula of
Theorem~\ref{thm:main} and the quadratic transition of
Theorem~\ref{thm:quadratic} also hold for nonlinear repair over
a prime base field.

More generally, if $q=p^e$ and an RS code over $\F_{q^\ell}$
attains the cut-set bound by nonlinear repair for every repair
instance, then
\begin{equation}\label{eq:nl-general-base-lower}
 e\ell\ge sQ_s(n)
\end{equation}
whenever $n\le2d-2$, or whenever
$n\le(1-\varepsilon)d^2$ and $d\ge d_0(\varepsilon)$ for a
fixed $0<\varepsilon<1$.
\end{corollary}
\begin{proof}
Apply Theorem~\ref{thm:nonlinear} to each repair instance.
The same evaluation set, viewed in $\F_{p^{e\ell}}$, admits
$\F_p$-linear repair with $e\ell/s$ prime-field symbols from
each helper. When $e=1$, this gives
$L_p^{\mathrm{NL}}\ge L_p$; the reverse inequality holds
because linear maps are allowed among nonlinear schemes.
This proves \eqref{eq:nl-prime-equality} and transfers both
sides of Theorem~\ref{thm:quadratic}.
For general $e$, apply Theorems~\ref{thm:main} and
\ref{thm:quadratic} with base field $\F_p$ to obtain
\eqref{eq:nl-general-base-lower}.
\end{proof}

For a nonprime base field, \eqref{eq:nl-general-base-lower}
controls the prime-field degree $e\ell$.
The divisibility conclusion is $s\mid e\ell$, which need not
imply $s\mid\ell$.

\section{Conclusion}\label{sec:conclusion}

The first version determined the exact minimum extension degree
for $n\le2d-2$, characterized the admissible full-helper degrees,
and bounded the number of coordinates admitting repair at the
cut-set bound. It also proved prime-field linearization of
nonlinear repair at equality.
The matching conclusions in the overlapping ranges of
\cite{FQXF} provide additional corroboration; that work also
gives a larger explicit equality range for $d\ge6$.

The full-helper coordinate-count bound is attained when
$s\mid\ell$ and $q\ge n-2$.
The coprime-product formula remains exact up to
$(1-\varepsilon)d^2$ for sufficiently large $d$, uniformly over
the base field and code dimension. Combining this lower bound
with the affine-plane upper bound of~\cite{FQXF} places the
first strict improvement at $(1+o(1))d^2$.
Thus $d^2$ is the asymptotic location of the transition.

The exact boundary in the window $n=d^2+o(d^2)$ remains open,
as does the minimum extension degree in the strict-improvement
region. Another question is whether
$L_q(n,k,d)=s\Delta_s(n,d)$ holds for all admissible parameters,
as conjectured in~\cite{FQXF}.

For $d<n-1$, a further direction is to count the pairs consisting
of a failed coordinate and a $d$-element helper set that admit
repair at the cut-set bound. For fixed $q,\ell,n,k,d$, define
\[
 \begin{aligned}
 N_{\rm cs}^{(d)}(A)=\bigl|\bigl\{(i,H):{}&
 i\in A,\ H\subseteq A\setminus\{i\},\ |H|=d,\\
 &\beta_A(i;H)=\frac{d\ell}{d-k+1}\bigr\}\bigr|.
 \end{aligned}
\]
Determining the maximum of $N_{\rm cs}^{(d)}(A)$ over evaluation
sets $A\subseteq\F_{q^\ell}$ of size $n$, and constructing codes
that attain it, would extend the coordinate-count problem studied
here for $d=n-1$.

\appendix
\renewcommand{\thesection}{\Alph{section}}
\section{Proof of the Finite-Parameter Bound}\label{app:finite}

This appendix proves Theorem~\ref{thm:main} using the notation
$t,S_i,N,Z,z$ and the prime levels introduced in Section~\ref{sec:lower}.
The helper-set intersection condition gives the required
prime-divisor count; an arithmetic comparison completes the proof.

\begin{lemma}\label{lem:fixed-target}
Suppose $n\le2d-2$. Fix $i\in A$, a prime $\rho\mid t$,
and indices $1\le j\le j'\le a_\rho$.
If $D$ and $D'$ are $E_{\rho,j}$- and
$E_{\rho,j'}$-sublines, respectively, such that
\[
 |A\cap D|\ge d,\qquad |A\cap D'|\ge d,
 \qquad i\notin D\cup D',
\]
then $D\subseteq D'$.
\end{lemma}

\begin{proof}
Both $A\cap D$ and $A\cap D'$ are subsets of $A\setminus\{i\}$.
Consequently,
\[
 |A\cap D\cap D'|
 \ge |A\cap D|+|A\cap D'|-(n-1)
 \ge2d-n+1\ge3.
\]
Since $E_{\rho,j}\subseteq E_{\rho,j'}$,
Lemma~\ref{lem:subline-intersection} gives $D\subseteq D'$.
\end{proof}

\begin{lemma}
\label{lem:block}
Suppose $n\le2d-2$ and $z>0$. Then
\begin{equation}\label{eq:block-bound}
 \omega(N)\ge\max\{d+1,\ \min\{z,5\}(d-3)\}.
\end{equation}
\end{lemma}

\begin{proof}
Fix $i\in Z$ and $B\subseteq A\setminus\{i\}$ of size
$f\le d$. Define the following set of prime divisors of $t$:
\[
 \mathcal R_i(B)=\left\{\rho\mid t:
 \begin{array}{l}
 \exists\,1\le j\le a_\rho,\ \exists\,E_{\rho,j}\text{-subline }D:\\
 |A\cap D|\ge d,\quad |A\setminus D|\ge2,\\
 B\subseteq D,\quad i\notin D
 \end{array}\right\}.
\]
We first prove
\begin{equation}\label{eq:prime-core-count}
 |\mathcal R_i(B)|\ge d-f+1.
\end{equation}
Taking $B=\varnothing$ will give the first bound in the lemma.
For the second bound, we will apply this inequality with
$B_i=G\setminus\{i\}$ inside a five-point set $G$ and show that
the resulting sets $\mathcal R_i(B_i)$ are pairwise disjoint.
For $\rho\in\mathcal R_i(B)$, at least one pair $(\rho,j)$
satisfies the hypotheses of \eqref{eq:subline-N}. Hence
$(\rho,j)\in\mathcal N$ and
$\rho\mid\prod_{(\rho',j')\in\mathcal N}\rho'=N$.
For each $\rho\in\mathcal R_i(B)$, let $j_\rho$ be the
largest index $j$ for which a subline in this definition
exists, and choose such a subline $D_\rho$ over
$E_{\rho,j_\rho}$. For any $E_{\rho,j}$-subline $D$ satisfying
the conditions defining $\mathcal R_i(B)$, maximality gives
$j\le j_\rho$, and Lemma~\ref{lem:fixed-target} with
$j'=j_\rho$ gives $D\subseteq D_\rho$.
Choose
\[
 x_\rho\in A\setminus(D_\rho\cup\{i\}).
\]
Such a point exists because $i\notin D_\rho$ and
$|A\setminus D_\rho|\ge2$.
If $|\mathcal R_i(B)|\le d-f$, then
$|B\cup\{x_\rho:\rho\in\mathcal R_i(B)\}|\le d$.
Choose a $d$-element helper set $H$ containing this union:
\[
 B\cup\{x_\rho:\rho\in\mathcal R_i(B)\}
 \subseteq H\subseteq A\setminus\{i\},\qquad |H|=d.
\]
This is possible since the union excludes $i$ and $d\le n-1$.
Suppose that
$j\in\mathcal J_\rho(H,i)$ and
$D=D_{\rho,j}(H)$ satisfies $|A\setminus D|\ge2$. Then
\[
 B\subseteq H\subseteq D,\qquad i\notin D.
\]
Thus $\rho\in\mathcal R_i(B)$ and $D\subseteq D_\rho$,
contradicting $x_\rho\in H\setminus D_\rho$.
Consequently, for every $\rho\mid t$ and
$j\in\mathcal J_\rho(H,i)$,
\[
 A\setminus D_{\rho,j}(H)=\{i\},\qquad
 (\rho,j)\in\mathcal S_i.
\]
The product formulas \eqref{eq:degree-level-product} and
\eqref{eq:level-sets} therefore give
\[
 [T(H\cup\{i\}):T(H)]\mid S_i<s,
\]
contrary to Lemma~\ref{lem:field}. This proves
\eqref{eq:prime-core-count} for every such $i$ and $B$.
Taking $B=\varnothing$ gives
$\omega(N)\ge d+1$.

For the other bound, let $h=\min\{z,5\}$ and choose
$I\subseteq Z$ with $|I|=h$. Since $n\ge5$, extend $I$
to a five-point set $G\subseteq A$. For each $i\in I$,
put $B_i=G\setminus\{i\}$. Equation~\eqref{eq:prime-core-count}
gives $|\mathcal R_i(B_i)|\ge d-3$.
If a prime $\rho$ belonged to both $\mathcal R_i(B_i)$ and
$\mathcal R_{i'}(B_{i'})$ for distinct $i,i'\in I$,
there would exist indices $j_i,j_{i'}$ and sublines
$D_i,D_{i'}$ over $E_{\rho,j_i},E_{\rho,j_{i'}}$, respectively,
such that
\[
 G\setminus\{i\}\subseteq D_i,\quad i\notin D_i,
 \qquad
 G\setminus\{i'\}\subseteq D_{i'},\quad i'\notin D_{i'}.
\]
These sublines contain the same three points in
$G\setminus\{i,i'\}$. The two fields are nested by
\eqref{eq:levels}.
Lemma~\ref{lem:subline-intersection} therefore gives
$D_i\subseteq D_{i'}$ or $D_{i'}\subseteq D_i$. However,
\begin{equation}\label{eq:incompatible-inclusions}
 i'\in D_i,\quad i'\notin D_{i'},
 \qquad i\in D_{i'},\quad i\notin D_i,
\end{equation}
which excludes either inclusion. The prime sets are disjoint,
and \eqref{eq:prime-core-count} yields
\[
 \omega(N)\ge\sum_{i\in I}|\mathcal R_i(B_i)|
 \ge h(d-3).
\]
\end{proof}

\subsection{Completion of the Lower Bound}

This subsection completes the lower bound
$\ell\ge sQ_s(n)$ by converting the prime-divisor count into
the inequality $t\ge Q_s(n)$.

\begin{lemma}\label{lem:arithmetic}
For $d\ge5$ and $1\le z\le4$,
\begin{equation}\label{eq:prime-compare}
 \Pi_{\max\{d+1,z(d-3)\}}>p_{3d-6}^{z}.
\end{equation}
\end{lemma}

\begin{proof}
For $d\ge7$, we first prove
\[
 \Pi_{d-3}>p_{3d-6}.
\]
The initial case is $\Pi_4=210>p_{15}=47$.
Bertrand's postulate gives $p_{j+3}<8p_j$.
Increasing $d$ by one multiplies the left side by
$p_{d-2}\ge11>8$, which proves the claim by induction.
Since $p_{r+a}\ge p_a$ for all $r\ge0$ and $a\ge1$,
\[
 \Pi_{\max\{d+1,z(d-3)\}}
 \ge\Pi_{z(d-3)}
 =\prod_{r=0}^{z-1}\prod_{a=1}^{d-3}p_{r(d-3)+a}
 \ge\Pi_{d-3}^{z}>p_{3d-6}^{z}.
\]
For $d=5,6$, note that $p_9=23$, $p_{12}=37$, and
\[
 \Pi_6=30030>23^3,\qquad
 \Pi_7=510510>37^3,\qquad
 \Pi_8=9699690>37^4.
\]
The first two comparisons cover $1\le z\le3$.
For $z=4$, the index $\max\{d+1,4(d-3)\}$ is at least
eight, so the last comparison proves \eqref{eq:prime-compare}.
\end{proof}

The comparison in \eqref{eq:prime-compare} applies when $d\ge5$.
For $d=4$, the range $3\le k<d<n\le2d-2$ allows only
$(n,k,d)=(5,3,4)$ and $(6,3,4)$. Theorem~\ref{thm:all}
settles the first case; the following lemma treats the second.

\begin{lemma}\label{lem:six}
Let $A\subseteq\Pone(\F_{q^t})$ satisfy
$|A|=6$ and $T(A)=\F_{q^t}$. Suppose that
\[
 [T(H\cup\{i\}):T(H)]\ge2
 \quad(i\in A,\ H\subseteq A\setminus\{i\},\ |H|=4).
\]
Then $t\ge\Pi_6=Q_2(6)$.
\end{lemma}

\begin{proof}
For $C\subseteq A$ with $|C|\ge3$, write $t_C=[T(C):K]$.
Define
\[
 \begin{aligned}
 \mathcal I
 &=\{(i,H):i\in A,\ H\subseteq A\setminus\{i\},\ |H|=4\},\\
 \delta_{i,H}
 &=[T(H\cup\{i\}):T(H)]=t_{H\cup\{i\}}/t_H,\\
 \mathcal I_\rho
 &=\{(i,H)\in\mathcal I:\rho\mid\delta_{i,H}\}
 \qquad(\rho\mid t).
 \end{aligned}
\]
Thus $|\mathcal I|=6\binom54=30$. For each five-element
subset $C\subset A$ and distinct $i,i'\in C$,
\eqref{eq:delete} and \eqref{eq:fields} give
\[
 \begin{gathered}
 t_C=\lcm(t_{C\setminus\{i\}},t_{C\setminus\{i'\}}),\\
 \gcd\!\left(\frac{t_C}{t_{C\setminus\{i\}}},
             \frac{t_C}{t_{C\setminus\{i'\}}}\right)=1.
 \end{gathered}
\]
For a fixed prime $\rho\mid t$, at most one of these five
ratios is divisible by $\rho$. Since $H=C\setminus\{i\}$
when $H\cup\{i\}=C$, we obtain
\[
 |\mathcal I_\rho|
 =\sum_{\substack{C\subset A\\|C|=5}}
   \left|\left\{i\in C:
       \rho\mid\frac{t_C}{t_{C\setminus\{i\}}}\right\}\right|
 \le\binom65=6.
\]
By hypothesis, $\delta_{i,H}\ge2$ for every
$(i,H)\in\mathcal I$. Also
$\delta_{i,H}\mid t_{H\cup\{i\}}\mid t$. Hence each
$\delta_{i,H}$ has a prime divisor $\rho\mid t$, and
$\mathcal I=\bigcup_{\rho\mid t}\mathcal I_\rho$.
Consequently,
\begin{equation}\label{eq:six-count}
 \begin{aligned}
  30=|\mathcal I|
  &\le\sum_{(i,H)\in\mathcal I}\omega(\delta_{i,H})
   =\sum_{\rho\mid t}|\mathcal I_\rho|\\
  &\le6\omega(t),\qquad \omega(t)\ge5.
 \end{aligned}
\end{equation}

We next show that $|\mathcal I_\rho|\le5$ whenever $\rho$
is odd and $v_\rho(t)=1$. Set
\[
 L=\F_{q^t},\qquad E=\F_{q^{t/\rho}},
\]
and define the following set of sublines:
\[
 \mathcal D_E
 =\{D=g(\Pone(E)):g\in\PGL_2(L),\ |A\cap D|\ge4\}.
\]
The set $\mathcal D_E$ may be empty. No member contains all
of $A$, since \eqref{eq:subline-field} would imply
$T(A)\subseteq E$, contrary to $T(A)=\F_{q^t}$.
Thus $|A\cap D|\in\{4,5\}$ for every $D\in\mathcal D_E$.
Since $a_\rho=1$ and $E_{\rho,1}=E$, \eqref{eq:witness} gives
\[
 \mathcal I_\rho
 =\bigcup_{D\in\mathcal D_E}
   \{(i,H):i\in A\setminus D,\ H\subseteq A\cap D,\ |H|=4\}.
\]
This union is disjoint: a pair $(i,H)$ occurring for $D$ and
$D'$ would give $|D\cap D'|\ge|H|=4$, forcing $D=D'$ by
Lemma~\ref{lem:subline-intersection}. Therefore
\begin{equation}\label{eq:six-subline-count}
 |\mathcal I_\rho|
 =\sum_{D\in\mathcal D_E}
   (6-|A\cap D|)\binom{|A\cap D|}{4}.
\end{equation}
In particular, $\mathcal D_E=\varnothing$ implies
$|\mathcal I_\rho|=0$.

If there is a $D_0\in\mathcal D_E$ with $|A\cap D_0|=5$,
then every $D\in\mathcal D_E$ satisfies
\[
 |A\cap D\cap D_0|\ge |A\cap D|+|A\cap D_0|-6\ge3.
\]
Hence $D=D_0$, so $\mathcal D_E=\{D_0\}$ and
\eqref{eq:six-subline-count} gives $|\mathcal I_\rho|=5$.

If no such $D_0$ exists, write
$\mathcal D_E=\{D_1,\ldots,D_r\}$ and put
$B_j=A\cap D_j$, $C_j=A\setminus D_j$. Then
$|B_j|=4$ and $|C_j|=2$. For $j\ne j'$, counting in $A$
and Lemma~\ref{lem:subline-intersection} give
\[
 2=4+4-6\le|B_j\cap B_{j'}|\le2,\qquad
 |C_j\cap C_{j'}|
 =6-|B_j\cup B_{j'}|=0.
\]
The sets $C_1,\ldots,C_r$ are pairwise disjoint subsets of
the six-element set $A$, so $2r\le6$. The case $r=3$ is
excluded by Lemma~\ref{lem:six-geometry}, because
$[L:E]=\rho$ is an odd prime. Thus $r\le2$ and
\eqref{eq:six-subline-count} gives $|\mathcal I_\rho|=2r\le4$.
We have proved $|\mathcal I_\rho|\le5$ in all cases.

If $\omega(t)=5$, both inequalities in \eqref{eq:six-count}
are equalities, so $|\mathcal I_\rho|=6$ for every prime
$\rho\mid t$. All odd prime divisors must therefore have
exponent at least two. Write the five distinct prime divisors
as $\rho_1<\cdots<\rho_5$. Since $\rho_j$ is odd for
$j\ge2$ and $\rho_j\ge p_j$,
\[
 t\ge\rho_1\prod_{j=2}^5\rho_j^2
 \ge2(3\cdot5\cdot7\cdot11)^2>\Pi_6.
\]
If $\omega(t)\ge6$, then $t\ge\Pi_6$ directly.
The identity $Q_2(6)=\Pi_6$ follows from
Proposition~\ref{prop:prime-packing}.
\end{proof}

\begin{proof}[Proof of Theorem~\ref{thm:main}]
Corollary~\ref{cor:upper} gives the upper bound.
The case $n=d+1$ is Theorem~\ref{thm:all}, so assume
$n\ge d+2$. For the lower bound, \eqref{eq:intrinsic}
reduces the claim to $t\ge Q_s(n)$.
If $d=4$, the only remaining parameters are $(n,k,d)=(6,3,4)$.
The normalized set $A$ satisfies the hypotheses of
Lemma~\ref{lem:six} by \eqref{eq:local}, which gives
$t\ge Q_2(6)$.

Assume $d\ge5$. Since $s\le d-2$, every prime at most $s-1$
lies in $\{2,\ldots,d-3\}$, a set of $d-4$ integers.
Together with $n\le2d-2$, this gives
\begin{equation}\label{eq:index-bound}
 c_s=\pi(s-1)\le d-4,\qquad n+c_s\le3d-6.
\end{equation}
If $z=0$, \eqref{eq:factor-bound} proves the claim.
If $1\le z\le4$, Lemma~\ref{lem:block},
\eqref{eq:prime-compare}, and \eqref{eq:index-bound} give
\[
 N\ge\Pi_{\max\{d+1,z(d-3)\}}
   >p_{3d-6}^{z}\ge p_{n+c_s}^{z}.
\]
Together with \eqref{eq:factor-bound} and \eqref{eq:Qgrowth},
this yields
\[
 t\ge Q_s(n-z)N
   >Q_s(n-z)p_{n+c_s}^{z}\ge Q_s(n).
\]
If $z\ge5$, \eqref{eq:block-bound} and \eqref{eq:index-bound} give
$\omega(N)\ge5(d-3)\ge3d-6\ge n+c_s$.
Thus
\[
 t\ge N\ge\Pi_{n+c_s}\ge Q_s(n)
\]
by \eqref{eq:Qgrowth}. This completes the lower bound.
\end{proof}


\begin{thebibliography}{99}

\bibitem{DGW}
A.~G. Dimakis, P.~B. Godfrey, Y.~Wu, M.~J. Wainwright,
and K.~Ramchandran,
``Network coding for distributed storage systems,''
\emph{IEEE Trans. Inf. Theory}, vol.~56, no.~9,
pp.~4539--4551, Sep.~2010,
doi: \href{https://doi.org/10.1109/TIT.2010.2054295}{10.1109/TIT.2010.2054295}.

\bibitem{GW}
V.~Guruswami and M.~Wootters,
``Repairing Reed--Solomon codes,''
\emph{IEEE Trans. Inf. Theory}, vol.~63, no.~9,
pp.~5684--5698, Sep.~2017,
doi: \href{https://doi.org/10.1109/TIT.2017.2702660}{10.1109/TIT.2017.2702660}.

\bibitem{YBISIT}
M.~Ye and A.~Barg,
``Explicit constructions of MDS array codes and RS codes with
optimal repair bandwidth,''
in \emph{Proc. IEEE Int. Symp. Inf. Theory (ISIT)},
Barcelona, Spain, Jul.~2016, pp.~1202--1206,
doi: \href{https://doi.org/10.1109/ISIT.2016.7541489}{10.1109/ISIT.2016.7541489}.

\bibitem{DM}
H.~Dau and O.~Milenkovic,
``Optimal repair schemes for some families of full-length
Reed--Solomon codes,''
in \emph{Proc. IEEE Int. Symp. Inf. Theory (ISIT)},
Aachen, Germany, Jun.~2017, pp.~346--350,
doi: \href{https://doi.org/10.1109/ISIT.2017.8006547}{10.1109/ISIT.2017.8006547}.

\bibitem{TYB}
I.~Tamo, M.~Ye, and A.~Barg,
``The repair problem for Reed--Solomon codes: Optimal repair
of single and multiple erasures with almost optimal node size,''
\emph{IEEE Trans. Inf. Theory}, vol.~65, no.~5,
pp.~2673--2695, May~2019,
doi: \href{https://doi.org/10.1109/TIT.2018.2884075}{10.1109/TIT.2018.2884075}.

\bibitem{ZFF}
W.~Zhao, F.-W.~Fu, and X.~Fu,
``A partial-exclusion repair scheme for MDS codes,''
arXiv:2603.12585v1, Mar.~13, 2026.
[Online]. Available: \url{https://arxiv.org/abs/2603.12585v1}.

\bibitem{QFFXISIT}
J.~Qiu, W.~Fang, S.-T.~Xia, and F.-W.~Fu,
``Improved constructions of Reed--Solomon codes with optimal
repair bandwidth,''
in \emph{Proc. IEEE Int. Symp. Inf. Theory (ISIT)},
Guangzhou, China, Jun.--Jul.~2026, pp.~1--6,
doi: \href{https://doi.org/10.1109/ISIT62367.2026.11653944}{10.1109/ISIT62367.2026.11653944}.

\bibitem{QFFX}
J.~Qiu, W.~Fang, S.-T.~Xia, and F.-W.~Fu,
``Reed--Solomon codes with optimal repair bandwidth:
A basis-transformation approach,''
arXiv:2601.10685v4, Oct.~17, 2026.
[Online]. Available: \url{https://arxiv.org/abs/2601.10685v4}.

\bibitem{LWJ}
W.~Li, Z.~Wang, and H.~Jafarkhani,
``On the sub-packetization size and the repair bandwidth of
Reed--Solomon codes,''
\emph{IEEE Trans. Inf. Theory}, vol.~65, no.~9,
pp.~5484--5502, Sep.~2019,
doi: \href{https://doi.org/10.1109/TIT.2019.2917425}{10.1109/TIT.2019.2917425}.

\bibitem{GJ}
V.~Guruswami and H.~Jiang,
``Near-optimal repair of Reed--Solomon codes with low
sub-packetization,''
in \emph{Proc. IEEE Int. Symp. Inf. Theory (ISIT)},
Paris, France, Jul.~2019, pp.~1077--1081,
doi: \href{https://doi.org/10.1109/ISIT.2019.8849719}{10.1109/ISIT.2019.8849719}.

\bibitem{CV}
A.~Chowdhury and A.~Vardy,
``Improved schemes for asymptotically optimal repair of MDS codes,''
\emph{IEEE Trans. Inf. Theory}, vol.~67, no.~8,
pp.~5051--5068, Aug.~2021,
doi: \href{https://doi.org/10.1109/TIT.2021.3071447}{10.1109/TIT.2021.3071447}.

\bibitem{WLLZZ}
K.~Wang, Z.~Liu, R.~Li, Y.~Zhao, and Y.~Zhang,
``Asymptotically optimal repair of Reed--Solomon codes with small
sub-packetization under rack-aware model,''
in \emph{Proc. IEEE Inf. Theory Workshop (ITW)},
Sydney, Australia, Sep.--Oct.~2025, pp.~674--679,
doi: \href{https://doi.org/10.1109/ITW62417.2025.11240535}{10.1109/ITW62417.2025.11240535}.

\bibitem{KKK}
W.~Kim, S.~Kruglik, and H.~M.~Kiah,
``Trace repair never loses to classical repair,''
in \emph{Proc. IEEE Int. Symp. Inf. Theory (ISIT)},
Guangzhou, China, Jun.--Jul.~2026, pp.~1--6,
doi: \href{https://doi.org/10.1109/ISIT62367.2026.11653992}{10.1109/ISIT62367.2026.11653992}.

\bibitem{LXZ}
Z.~Liu, J.~Xu, and Z.~Zhang,
``Calculating the I/O cost of linear repair schemes for RS codes
evaluated on subspaces via exponential sums,''
\emph{IEEE Trans. Inf. Theory}, vol.~72, no.~2,
pp.~994--1010, Feb.~2026,
doi: \href{https://doi.org/10.1109/TIT.2025.3645309}{10.1109/TIT.2025.3645309}.

\bibitem{DD}
I.~Duursma and H.~Dau,
``Low bandwidth repair of the RS(10,4) Reed--Solomon code,''
in \emph{Proc. Inf. Theory Appl. Workshop (ITA)},
San Diego, CA, USA, Feb.~2017,
doi: \href{https://doi.org/10.1109/ITA.2017.8023449}{10.1109/ITA.2017.8023449}.

\bibitem{AR}
D.~Aggarwal and S.~Ram,
``Polynomial matrices, splitting subspaces and Krylov subspaces
over finite fields,''
\emph{Finite Fields Appl.}, vol.~83, Art.~no.~102081, Oct.~2022,
doi: \href{https://doi.org/10.1016/j.ffa.2022.102081}{10.1016/j.ffa.2022.102081}.

\bibitem{CT}
R.~Con and I.~Tamo,
``Nonlinear repair of Reed--Solomon codes,''
\emph{IEEE Trans. Inf. Theory}, vol.~68, no.~8,
pp.~5165--5177, Aug.~2022,
doi: \href{https://doi.org/10.1109/TIT.2022.3167615}{10.1109/TIT.2022.3167615}.

\bibitem{LN}
R.~Lidl and H.~Niederreiter,
\emph{Finite Fields}, 2nd ed., Encyclopedia of Mathematics
and its Applications, vol.~20.
Cambridge, U.K.: Cambridge University Press, 1997.

\bibitem{Cameron}
P.~J. Cameron,
\emph{Projective and Polar Spaces}, 2nd ed., lecture notes.
London, U.K.: Queen Mary and Westfield College, 2000.
[Online]. Available:
\url{https://webspace.maths.qmul.ac.uk/p.j.cameron/pps/}.

\bibitem{KS}
W.~H. Kautz and R.~C. Singleton,
``Nonrandom binary superimposed codes,''
\emph{IEEE Trans. Inf. Theory}, vol.~10, no.~4,
pp.~363--377, Oct.~1964,
doi: \href{https://doi.org/10.1109/TIT.1964.1053689}{10.1109/TIT.1964.1053689}.

\bibitem{FirstVersion}
Y.~Li,
``Exact minimum extension degrees for Reed--Solomon codes under
linear and nonlinear repair at the cut-set bound,''
arXiv:2610.34377v1, Oct.~25, 2026.
\url{https://arxiv.org/abs/2610.34377v1}.

\bibitem{FQXF}
W.~Fang, J.~Qiu, S.-T.~Xia, and F.-W.~Fu,
``The minimum subpacketization of Reed--Solomon codes with linear
repair achieving the cut-set bound,''
arXiv:2610.37522v1, Oct.~27, 2026.
\url{https://arxiv.org/abs/2610.37522v1}.

\bibitem{HZ}
Y.~Hu and L.~Zhang,
``A proof of the strong Erd\H{o}s--Frankl--F\"uredi cover-free
conjecture,'' arXiv:2610.24442v1, 2026.
\url{https://arxiv.org/abs/2610.24442v1}.

\bibitem{Dusart}
P.~Dusart,
``Estimates of some functions over primes without R.H.,''
arXiv:1002.0442, 2010.
\url{https://arxiv.org/abs/1002.0442}.

\end{thebibliography}
\end{document}